\documentclass[11pt,onecolumn]{article}
\usepackage[top=1in, bottom=1in, left=1.25in, right=1.25in]{geometry}
\usepackage{appendix}
\usepackage{amsfonts}
\usepackage{amsmath}
\usepackage{amssymb}
\usepackage{graphicx}
\usepackage{color, soul}
\usepackage{stfloats}

\usepackage[amsmath,thmmarks]{ntheorem}
\usepackage{theorem}
\newtheorem{theorem}{Theorem}

\newtheorem{proposition}{Proposition}

\theoremheaderfont{\sc}\theorembodyfont{\upshape}
\theoremstyle{nonumberplain}
\theoremseparator{}
\theoremsymbol{\rule{1ex}{1ex}}
\newtheorem{proof}{Proof}

\begin{document}
\title{Maximum Likelihood Estimation from Sign Measurements with Sensing Matrix Perturbation}

\author{Jiang~Zhu, Xiaohan~Wang, and Yuantao~Gu% <-this % stops a space
\thanks{The authors are with the Department of Electronic Engineering, Tsinghua University, Beijing 100084, CHINA. The corresponding author of this work is Yuantao Gu (e-mail: gyt@tsinghua.edu.cn).}% <-this % stops a space
}
%\thanks{Manuscript received Dec.~28, 2011. This work was supported partially by the National Natural Science Foundation of CHINA (60872087 and U0835003) and Agilent Technologies Foundation \# 2205.}

\date{December 6, 2013}

\maketitle

\begin{abstract}
%\boldmath
The problem of estimating an unknown deterministic parameter vector from sign measurements with a perturbed sensing matrix is studied in this paper. We analyze the best achievable mean square error (MSE) performance by exploring the corresponding Cram\'{e}r-Rao Lower Bound (CRLB). To estimate the parameter, the maximum likelihood (ML) estimator is utilized and its consistency is proved. We show that the perturbation on the sensing matrix exacerbates the performance of ML estimator in most cases. However, suitable perturbation may improve the performance in some special cases. Then we reformulate the original ML estimation problem as a convex optimization problem, which can be solved efficiently. Furthermore, theoretical analysis implies that the perturbation-ignored estimation is a scaled version with the same direction of the ML estimation. Finally, numerical simulations are performed to validate our theoretical analysis.

\end{abstract}
%\begin{IEEEkeywords}
%Maximum likelihood estimation, sign measurements, Gaussian perturbation, CRLB.
%\end{IEEEkeywords}

%\ifCLASSOPTIONpeerreview
%\begin{center} \bfseries EDICS Category: DSP-RECO \end{center}
%\fi

%\IEEEpeerreviewmaketitle
\section{Introduction}
The linear regression problem with perturbed sensing matrix has been extensively studied in recent years \cite{Beck, Minimax, Compare_model_unc}. Mathematically, the vector ${\mathbf y}\in {\mathbb R}^N$ is observed via a corrupted sensing matrix as
\begin{align}\label{Pertur_lin_model}
     {\mathbf y}=({\mathbf H}+{\mathbf E})^{\rm T}{\mathbf w}+{\mathbf n},
\end{align}
where ${\mathbf H}\in {\mathbb R}^{p\times N}$ is a deterministic known sensing matrix, and $\mathbf E$ is a random matrix each of whose elements is \emph{i.i.d.}, $e_{ij}\sim \mathcal{N}(0,\sigma_e^2), i=1,\cdots, p, j=1,\cdots,N$. The additive noise vector $\mathbf n$ is independent of $\mathbf E$ and satisfies ${{\mathbf n}\sim {\mathcal N}({\mathbf 0},\sigma_n^2{\mathbf I})}$, where $\sigma_e^2$ is viewed as the strength of perturbation. To estimate the unknown parameter vector ${\mathbf w}\in {\mathbb R}^p$, the perturbation $\mathbf E$ is treated as a nuisance parameter and the maximum likelihood (ML) method is used. Several numerical methods have been proposed including minimax search, maximin search, and the classical expectation-maximization (EM) algorithm \cite{Compare_model_unc}.

It is natural to further study the parameter estimation problem with perturbed sensing matrix by the sign measurements
\begin{align}
     {\mathbf y}={\rm sign}\left(\left({\mathbf H}+{\mathbf E}\right)^{\rm T}\mathbf{w}+\mathbf{n}\right),\label{Gaussian_model}
\end{align}
where $\mathbf y$ denotes a binary measurement vector and ${\rm sign}(\cdot)$ is a vector each of whose entries is equal to the sign of the corresponding element (we assume that the sign of a real number is $1$ or $-1$, when the number is positive or nonpositive, respectively).

\subsection{Problem Background}

Most available works focus on the simplified case of (\ref{Gaussian_model}) in which the perturbation does not exist, i.e., ${\mathbf E}={\mathbf 0}$ \cite{Demaris, Conj_book}. In this setting, the model is reduced to
\begin{align}
     {\mathbf y}={\rm sign}\left({{\mathbf H}^{\rm T}}{\mathbf w}+{\mathbf n}\right).\label{Determin_matrix}
\end{align}
Model (\ref{Determin_matrix}) is closely related to the binary regression model in statistics, where only binary outcomes are obtained to estimate the factors that affect the results. When the noise is Gaussian distributed, the binary regression model is also called a probit model \cite{WNewey}, which can be described as
\begin{align}
     {\mathbf y}={\rm sign}\left({{\mathbf H}^{\rm T}}\tilde{\mathbf w}+\tilde{\mathbf n}\right),\label{Binary_reg}
\end{align}
where $\tilde{\mathbf n}$ is a normalized Gaussian vector satisfying $\tilde{\mathbf n} \sim \mathcal{N}({\mathbf 0},{\mathbf I})$, and $\tilde{\mathbf w}= {\mathbf w}/\sigma_n$ is what we wish to estimate. Once the estimation of $\tilde{\mathbf w}$ is acquired, the distribution of the sign measurement ${\rm sign}({\mathbf h}^{\rm T}\tilde{\mathbf w}+\tilde{\mathbf n})$ can be predicted for a new ${\mathbf h}\in {\mathbb R}^p$.

Another application related to model (\ref{Determin_matrix}) is to estimate some physical quantities (pressure, temperature, mean-location, and etc.) based on binary quantized measurements in wireless sensor network. The mathematical model of most related works in this scenario is a special case of (\ref{Determin_matrix}) in which the parameter to be estimated is a scalar. In this application, there are a large number of spatially distributed nodes. Each node is available to a subset of observations and has to transmit the information to the fusion center. Due to the limited bandwidth, the node may quantize the measurements coarsely. It is known that the minimum variance of the estimator based on binary measurements is only $\pi/2$ times of the clairvoyant estimator \cite{Papadopoulos, Oppen}, which motivates researchers to achieve this excellent performance by proposing carefully designed strategies. In \cite{Giannakis_conc, Giannakis1_bandwidth}, distributed estimation algorithms are proposed to reduce the transmission requirements by exploiting spatial correlation. Furthermore, a universal decentralized estimation scheme is proposed to cope with the unknown noise distribution case \cite{Luo, Luo1, Luo2}. While all above works focus on the estimation of the scalar case, \cite{Giannakis2_bandwidth} analyzes the performance of the ML estimator for multivariate parameters with dithered quantization.

\subsection{Main Contribution}

This paper focuses on the ML estimation of the vector parameter from sign measurements with sensing matrix perturbation. The main contribution of this work is two-fold. On the one hand, the Cram\'{e}r-Rao Lower Bound (CRLB) on the mean square error (MSE) is theoretically derived to analyze the performance of unbiased estimators. The ML estimator is proved to be consistent, then its performance is studied using the CRLB. It is shown that the perturbation on the sensing matrix worsens the performance in most cases. However, suitable perturbation may improve the estimation accuracy in some special cases. On the other hand, the ML estimation problem is reformulated as a convex optimization problem, implying that if the global optimal point exists, there are numerical algorithms guaranteed to converge to it. We analyze the probability that the optimal point of the ML estimator exists. It is shown that moderate perturbation may be beneficial by providing randomness for the measurements. Moreover, the mismodeling effects is studied in the case that the perturbation is ignored. We show that the estimator ignoring the perturbation can provide a scaled estimation with the same direction with that of the ML estimator. It implies that we can also obtain the correct direction estimation when the perturbation information is unknown. Finally, we compare the MSE performance of the ML estimator against the CRLB by simulation.

\subsection{Related Work}

For model (\ref{Determin_matrix}), there has been a lot of works focusing on the estimation of a scalar parameter \cite{Oppen, Giannakis1_bandwidth, Giannakis2_bandwidth}.
In \cite{Oppen}, the case in which the sensing matrix is ${\mathbf H}=[1,\cdots,1]$ is studied. The parameter $w$ are supposed to lie in the range $(-\Delta,\Delta)$. Thus the worst-case CRLB is optimized with respect to the variance of the additive noise. It is also shown that the performance of the estimation can be improved by a periodic waveform or feedback signal prior to quantization. Recently, an additive outlier $\mathbf o$ is introduced in (\ref{Determin_matrix})  to model the errors \cite{Rob_conj} by
\begin{align}
     {\mathbf y}={\rm {sign}}\left({{\mathbf H}^{\rm T}}\mathbf{w}+\mathbf{n}+\mathbf{o}\right).\label{Outlier_model}\nonumber
\end{align}
The sparsity of the outliers is controlled. Desirable tradeoff between model fit and complexity is attained by a new classification-based approach. In \cite{Tsakonas_SpaC_con, Tsakonas_SpaC}, both the outliers and the unknown parameters are sparse. The ML method for the probit model is proposed to estimate the model parameters. Suppose that the numbers of nonzero entries of $\mathbf o$ and $\mathbf w$ are less than or equal to $k_o$ and $k_w$, respectively. By defining the concatenated matrix ${\mathbf Q}\triangleq [{\mathbf H}^{\rm T}, {\mathbf I}_{N\times N}]$, a sufficient condition for the identifiability of $\mathbf w$ and $\mathbf o$ can be described by
\begin{align}
{\rm {Spark}}({\mathbf Q})>2(k_o+k_w),\notag
\end{align}
where ${\rm {Spark}}({\mathbf Q})$ denotes the minimum number of the dependent columns in $\mathbf Q$. The ML estimation of the vector $[{\mathbf w}^{\rm T},{\mathbf o}^{\rm T}]^{\rm T}$ is equivalent to the following optimization problem
\begin{align}
&\underset{{\mathbf w},{\mathbf o}}{\operatorname{minimize}}~-\left(\sum \limits_{i=1}^{N}{\log}\Phi \left(y_i\frac{{\mathbf h}_i^{\rm{T}}{\mathbf w}+o_i}{\sigma_n}\right)\right)\notag\\
&{\operatorname{subject~to}}~\|{\mathbf w}\|_0\leq k_w, \|{\mathbf o}\|_0\leq k_o,\nonumber
\end{align}
where $\Phi (u)=\frac{1}{\sqrt{2\pi}}\int_{-\infty}^{u}{{{\textrm e}^{-\frac{x^2}{2}}}}{\textrm d}x$ is the cumulative distribution function of the standard Gaussian distribution. In \cite{Tsakonas_SpaC}, it is shows that the outliers and the unknown parameters can be jointly estimated by using the convex $l_1$-norm to replace with the cardinality constraint. Though some methodologies utilized in above papers are adopted in this work, the model they studied is different from \eqref{Gaussian_model}.

For the probit model (\ref{Binary_reg}), the standard ML procedure is often used to estimate the unknown parameter vector. The ML estimation is equivalent to the following optimization problem
\begin{align}
\underset{\tilde{\mathbf w}}{\operatorname{minimize}}-\sum \limits_{i=1}^{N}{\log }\Phi \left(y_i{\mathbf h}_i^{\rm{T}}\tilde{\mathbf w}\right).\notag
\end{align}
This problem is convex and is first solved in \cite{Bliss}.
Model (\ref{Binary_reg}) with uncertainty in the sensing matrix has been studied in a number of literature. There are two approaches to describe the uncertainty of the sensing matrix \cite{Carroll}. The first approach is the standard errors in variables (EIV) model, where $\mathbf H$ is modeled as a deterministic unknown sensing matrix, and $\mathbf G$ is a noisy observation on $\mathbf H$ which can be described by ${\mathbf G}={\mathbf H}+{\mathbf E}$. Given the observations $\mathbf G$ and $\mathbf y$, both $\tilde{\mathbf w}$ and $\mathbf H$ are estimated by solving
\begin{align}
\underset{\tilde{\mathbf w},{\mathbf H}}{\operatorname{minimize}}~-l({{\mathbf y},{\mathbf G}};{\tilde{\mathbf w},{\mathbf H}}),\label{EIV_log}
\end{align}
where $l({{\mathbf y},{\mathbf G}};{\tilde{\mathbf w},{\mathbf H}})$ is the log-likelihood function of $\mathbf y$ and $\mathbf G$ parameterized by $\tilde{\mathbf w}$ and $\mathbf H$. Equation (\ref{EIV_log}) is equivalent to
\begin{align}\label{EIV_probit}
\underset{{\mathbf H},\tilde{\mathbf w}}{\operatorname{minimize}}\left(-\sum \limits_{i=1}^{N}{\log }\Phi \left(y_i{\mathbf h}_i^{\rm{T}}\tilde{\mathbf w}\right)+\frac{\|{\mathbf H}-{\mathbf G}\|_{\rm F}^2}{2\sigma_e^2}\right).
\end{align}
The number of variables increases by a factor of $p({\textstyle 1}+{\textstyle 1}/{\textstyle N})$ with respect to the number of measurements $N$. In \cite{Carroll}, it is shown that the ML estimator is in general not consistent, implying that the ML estimator will not converge to the true parameter in probability when the number of measurements tends to infinity. The second approach to describe the uncertainty is to model the sensing matrix as a random matrix. The statistical characterization of the sensing matrix is known, thus the nuisance parameter can be eliminated and the estimation of $\tilde{\mathbf w}$ is available. The above works focused on the regression analysis and some basic assumptions are different from this work.
\vspace{1em}

\leftline{\normalsize{\textbf Notation}}

For any scalar $x \in \mathbb R$, $\lfloor x \rfloor$ ($\lceil x \rceil$) denotes the nearest integer less than or equal to (greater than or equal to) $x$. For an unknown estimated parameter vector $\mathbf w$ (scalar parameter $w$), ${\mathbf w}_0$ ($w_0$) denotes its true value. For a random vector $\mathbf y$, ${\rm E}_{\mathbf y}[\cdot]$ denotes the expectation taken with respect to $\mathbf y$. For ${\mathbf w} = [w_1,\cdots,w_n]^{\rm T}$ and a continuous and differentiable function $f:\mathbb{R}^n\rightarrow \mathbb{R}$, $\nabla_{\mathbf w}f$ and $\nabla_{\mathbf w}^2f$ denotes its gradient and Hessian. For a vector function ${\mathbf g}:S \rightarrow \mathbb{R}^r$ defined on a set $S$ in $\mathbb{R}^s$, ${\partial{\mathbf g}({\boldsymbol \theta})}/{\partial{\boldsymbol \theta}}$ denotes its Jacobian matrix $[{\partial{g}_i({\theta})}/{\partial{\theta}_j}]_{r\times s}$. For any appropriate matrix $\mathbf A$, $a_{ij}$ denotes its ($i$,$j$)th element, ${\mathbf a}_i$ denotes its $i$th column, $\|\mathbf A\|_{\rm F}$ denotes its Frobenius norm, ${\rm {tr}}(\mathbf A)$ denotes its trace, ${\mathbf A}\succeq {\mathbf 0}$ (${\mathbf A}\succ {\mathbf 0}$) means that $\mathbf A$ is positive semidefinite (positive definite), and ${\mathbf A}\succeq {\mathbf B}$ means that ${\mathbf A}-{\mathbf B}\succeq {\mathbf 0}$. ${\rm diag}(\lambda_1,\cdots,\lambda_ p)$ is a $p\times p$ diagonal matrix with the $i$th diagonal elements $\lambda_i$. Other notations will be introduced when needed.

The rest of this paper is organized as follows. In Section II, the ML estimator is utilized and its consistency is proved. In Section III, the theoretical CRLB is derived, and the theoretical performance limits is analyzed. In Section IV, we reformulate the original ML estimation problem as a convex optimization problem. In Section V, we discuss the probability that the likelihood function is unimodal, and provide some insights on the similarity and the difference between the ML estimator and the perturbation-ignored estimator. In Section VI, the numerical results are presented. Finally we conclude the paper in Section VII.

\section{Maximum Likelihood Estimator}

The model (\ref{Gaussian_model}) can be written in a more canonical form
\begin{align}\label{Pertubation_model}
     {\mathbf y}={\rm {sign}}\left({\mathbf H}^{\rm T}{\mathbf w}+{\mathbf z}\right),
\end{align}
where ${\mathbf z}={{\mathbf E}^{\rm T}}{\mathbf w}+{\mathbf n}$ is regarded as the \emph{sum} of a multiplicative noise and an additive noise \cite{Meaerr_book}. The variance of the equivalent noise $\mathbf z$ depends on the parameter vector $\mathbf w$, which makes the problem more complex than the perturbation free setting. Because $e_{ij}$ is \emph{i.i.d.} Gaussian random variable, ${{\mathbf E}^{\rm T}}{\mathbf w}$ is an $N$ dimensional Gaussian distributed random vector. It follows by straightforward calculation that ${\rm E}[{{\mathbf E}^{\rm T}}{\mathbf w}]={\mathbf 0}$, and ${\rm {Cov}}[{{\mathbf E}^{\rm T}}{\mathbf w}] = \sigma_e^2\|{\mathbf w}\|_2^2{\mathbf I}$. Thus the variance of the multiplicative noise is $\sigma_e^2\|{\mathbf w}\|_2^2$. Then, from the mutual independence of ${{\mathbf E}^{\rm T}}{\mathbf w}$ and ${\mathbf n}$, one has ${\mathbf z}\sim {\mathcal N}({\mathbf 0},\sigma_z^2{\mathbf I})$, where
\begin{align}\label{Var_z}
     \sigma_z^2=\|\mathbf w\|_2^2\sigma_e^2+\sigma_n^2.
\end{align}

Now we calculate the likelihood function ${\rm Pr}({\mathbf y};{\mathbf w})$. Let ${\bf H}=[{\bf h}_1, {\bf h}_2, \cdots, {\bf h}_N]$, ${\mathcal I}_+$ and ${\mathcal I}_-$ denote the set of indices $\{i|y_i=1\}$ and $\{i|y_i=-1\}$, respectively. By partitioning the observations into ${\mathcal I}_+$ and ${\mathcal I}_-$, the likelihood function ${\rm Pr}(\mathbf y;\mathbf w)$ is calculated to be
\begin{align}\label{Likelihood}
     {\rm{Pr}}(\mathbf y;\mathbf w)=&\prod\limits_{i\in {\mathcal I}_+}{\rm{Pr}}({\mathbf h}_i^{\rm T}{\mathbf w}+z_i> 0)\prod\limits_{i\in {\mathcal I}_-}{\rm{Pr}}({\mathbf h}_i^{\rm T}{\mathbf w}+z_i\leq 0)\nonumber\\
     =&\prod\limits_{i\in {\mathcal I}_+}\Phi \left(\frac{{\mathbf h}_i^{\rm T}{\mathbf w}}{\sigma_z}\right)\prod\limits_{i\in {\mathcal I}_-}\Phi \left(-\frac{{\mathbf h}_i^{\rm T}{\mathbf w}}{\sigma_z}\right)\nonumber\\
     =&\prod\limits_{i=1}^{N}\Phi \left(y_i\frac{{\mathbf h}_i^{\rm T}{\mathbf w}}{\sigma_z}\right).\nonumber
\end{align}
The corresponding log-likelihood function $l({\mathbf y};{\mathbf w})$ is given by
\begin{align}
 l({\mathbf y};{\mathbf w})=\sum\limits_{i=1}^{N}{\log }\Phi \left(y_i\frac{{{\mathbf h}_i^{\rm T}{\mathbf w}}}{{\sigma_z}}\right).\label{Loglike}
\end{align}
Therefore, the ML estimation of the vector $\mathbf w$ is equivalent to minimizing the negative log-likelihood function (\ref{Loglike}). Substituting (\ref{Var_z}) in (\ref{Loglike}), one has
\begin{align}
%{\mathcal{P}}:
\underset{{\mathbf w}\in {\mathbb R}^p} {\operatorname{minimize}}~-\sum_{i=1}^N{\log }\Phi \left(y_i\frac{{{\mathbf h}_i^{\rm T}{\mathbf w}}}{{\sqrt{\|\mathbf w\|_2^2\sigma_e^2+\sigma_n^2}}}\right).\label{gen_p}
\end{align}

Now we briefly discuss the statistical identifiability of the model (\ref{Gaussian_model}). The model is statistically identifiable if the underlying parameter can be estimated accurately by an infinite number of measurements. Mathematically, this means that if ${\mathbf w}_1$ is not equal to ${\mathbf w}_2$, the corresponding measurements ${\mathbf y}_1$ and ${\mathbf y}_2$ must follow different probability distributions. A necessary and sufficient condition to guarantee the identifiability of model (\ref{Gaussian_model}) is that ${\mathbf H}$ should be of full row rank, which is the same with the linear regression model (\ref{Pertur_lin_model}).

We will close this section by studying the consistency of the ML estimator on
\begin{equation}\label{GeneralPertubation_model}
     y_i = {\rm {sign}}\left({\mathbf h}_i^{\rm T}{\mathbf w}+{z}_i\right), \quad i=1,2,\cdots,N,
\end{equation}
where ${\bf h}_i$ are generated from any underlying continuous distribution and ${ z}_i\sim {\mathcal N}({\mathbf 0},\sigma_z^2)$ is an \emph{i.i.d.} sequence.
The consistency means that as the number of the measurements $N$ tends to infinity, the estimator converges to the true parameter value ${\mathbf w}_0$ in probability. Though it has been demonstrated that the ML estimator (\ref{EIV_probit}) in EIV model is not consistent in general \cite{Carroll}, we could prove that the consistency of the ML estimator is satisfied in the model (\ref{GeneralPertubation_model}).
\begin{theorem}\label{ML_con}
Assume that $\mathbf w$ lies in the parameter space ${\mathcal W} = \{{\mathbf w}|\|\mathbf w\|_2\leq R_w\}$, where $R_w$ is a positive constant. $\{{\mathbf h}_i\}_{i=1}^N$ are generated from an underlying continuous distribution. The ML estimator (\ref{gen_p}) is consistent.
\end{theorem}
\begin{proof}
The proof is postponed to Appendix A.
\end{proof}

One may notice that the unknown parameter is assumed to be bounded, which is a technical mathematical condition needed for many theoretical analysis \cite{Compare_model_unc}. In practice, we can choose $R_w$ sufficiently large, then the estimator is assumed to have no knowledge of this constraint.

\section{Cram\'{e}r-Rao Lower Bound}

We now provide a lower bound on the variance of any unbiased estimator of the model (\ref{Gaussian_model}). It is well known that the MSE of the ML estimator asymptotically achieves the CRLB under certain regularity conditions. Therefore, the CRLB provides a reasonable benchmark shedding light on the performance of the ML estimator.

The Fisher information matrix (FIM) is used to find the bounds for unbiased estimators. We can calculate the FIM as the negative expectation of the Hessian of the log-likelihood function with respect to $\mathbf y$,
\begin{align}
{\mathbf J}({\mathbf w})=-{\rm E}_{\mathbf y}[\nabla_{\mathbf w}^2l({\mathbf y};{\mathbf w})].\notag
\end{align}
The CRLB matrix is equal to the inverse of the FIM by
\begin{align}
{\rm CRLB}({\mathbf w})= \left({\mathbf J}({\mathbf w})\right)^{-1},\notag
\end{align}
and the CRLB on the MSE is the trace of the CRLB matrix.

Now a closed-form expression of the CRLB on the MSE for the model (\ref{Gaussian_model}) is provided in the following theorem.
\begin{theorem} \label{crlb_theorem}
Consider the estimation of $\mathbf w$ in the model (\ref{Gaussian_model}) with both $\sigma_n^2$ and $\sigma_e^2$ known. The FIM is ${\mathbf J}(\mathbf{w})={\mathbf M}{\mathbf \Lambda}{\mathbf M}^{\rm T}$ and the MSE $ {\rm mse}({\hat{\mathbf w}})={\rm E}[\|{\hat{\mathbf w}}-{\mathbf w}\|_2^2]$ of any unbiased estimator ${\hat{\mathbf w}}$ satisfies
\begin{align} \label{mse_crlb}
{\rm {mse}}({\hat{\mathbf w}})\geq{\rm {tr}}\left(\left({\mathbf M}{\mathbf \Lambda}{\mathbf M}^{\rm T}\right)^{-1}\right),
\end{align}
where ${\mathbf \Lambda}$ is a positive diagonal matrix with elements
\begin{align}\label{lambda_ii}
    {\lambda}_{ii}= \frac{1}{{2\pi \sigma_z^2}}{\left( \frac{1}{{\Phi \left(\frac{\mathbf{h}_{i}^{\rm T}{\mathbf w}}{{\sigma_z}} \right)}}+\frac{1}{{\Phi \left(-\frac{{\mathbf h}_i^{\rm T}{\mathbf w}}{{\sigma_z}} \right)}} \right)}{{\rm e}^{-\frac{{{\left({\mathbf h}_{i}^{\rm T}{\mathbf w} \right)}^2}}{{\sigma_z^2}}}},
\end{align}
and
\begin{align}\label{Def_M}
    {\mathbf M}= \left({\mathbf I}-\frac{\sigma_e^2}{\sigma_z^2}{\mathbf w}{\mathbf w}^{\rm T}\right){\mathbf H}.
\end{align}
\end{theorem}
\begin{proof}
The proof is postponed to Appendix \ref{CRB_appendix}.
\end{proof}

For simplicity, let $\mathbf J$ denote the FIM instead of ${\mathbf J}(\mathbf{w})$ in the following text. Two extreme cases will be discussed. The first case corresponds to the setting of perturbation free. Then ${\mathbf M}={\mathbf H}$ and $\sigma_z^2= \sigma_n^2$. The FIM is degenerated to ${\mathbf J}= {\mathbf H}{\mathbf \Lambda}{\mathbf H}^{\rm T}$, which is consistent with \cite{Tsakonas_SpaC}. The second case corresponds to the setting of additive noise free. Hence $\mathbf M$ is rank deficient and ${\mathbf J}$ is singular, implying that there exists no finite variance unbiased estimator \cite{Singular_fim}. We can also see it in the reduced model
\begin{align}
{\mathbf y}={\rm {sign}}\left({{({\mathbf H}+{\mathbf E})}^{\rm T}}\mathbf{w}\right).\label{Spec_model2}
\end{align}
For an estimator $\hat{\mathbf w}$, its scaled version $k\hat{\mathbf w}$ satisfies (\ref{Spec_model2}) for all $k>0$. This result demonstrates that the magnitude information of $\mathbf w$ is lost from sign measurements. Therefore, the additive noise $\mathbf n$ is necessary for the estimation in that it provides a dynamic bias for the sign function \cite{MarkDaven}. We always assume that $\sigma_n^2$ is nonzero.

We will then discuss how the multiplicative noise and the additive noise affect the CRLB on the MSE. The sign measurement can be viewed as a nonlinear system, thus its performance can be enhanced by the presence of optimized random noise \cite{DeWeese, Douglass, Levin}. In the model (\ref{Gaussian_model}), there may exist optimal variances of multiplicative noise and additive noise that minimize the CRLB on the MSE. Viewing $\sigma_n^2$ and $\sigma_e^2$ as variables,  we will discuss three cases in the following subsections, corresponding to the situations in which the variance of the equivalent noise $\sigma_z^2$, of the multiplicative noise $\sigma_e^2\|{\mathbf w}\|_2^2$, or of the additive noise $\sigma_n^2$ is fixed, respectively.

\subsection{The Case of Equivalent Noise Fixed}
Suppose that we have two models. One is model (\ref{Gaussian_model}), the corresponding estimator and the FIM are denoted as $\hat{\mathbf w}({\mathbf y},{\mathbf H},\sigma_e^2,\sigma_n^2)$ and $\mathbf J$, respectively. The other is model (\ref{Determin_matrix}). We use $\hat{\mathbf w}({\mathbf y},{\mathbf H},0,\sigma_z^2)$ to denote the estimator with the FIM $\tilde{\mathbf J}$ in this situation. One may define $\gamma={\sigma_e^2\|{\mathbf w}\|_2^2}/{\sigma_n^2}$ to denote the ratio of the variance of the multiplicative noise to that of the additive noise. Let $\tilde{\lambda}_i$ denote the $i$th largest eigenvalue of the FIM $\tilde{\mathbf J}$, then the following result is obtained.
\begin{proposition}\label{Pro_ineq}
The multiplicative noise exacerbates the performance of estimation when the variance of equivalent noise is fixed. In the MSE sense, we have the following inequality
\begin{align}
\frac{{\gamma^2+2\gamma}}{\tilde{\lambda}_1}\leq{\rm {tr}}({\mathbf J}^{-1})-{\rm {tr}}(\tilde{\mathbf J}^{-1})\leq \frac{{\gamma^2+2\gamma}}{\tilde{\lambda}_p}.\label{inequ}
\end{align}
\end{proposition}
\begin{proof}
The proof is postponed to Appendix C.
\end{proof}

Proposition \ref{Pro_ineq} demonstrates that the minimum MSE is achieved at $\sigma_e^2=0$ when $\sigma_z^2$ is fixed. It is also shown that when the variance of the multiplicative noise is much smaller than that of the additive noise, the lower and the upper bounds of ${\rm {tr}}({\mathbf J}^{-1})-{\rm {tr}}(\tilde{\mathbf J}^{-1})$ are proportional to $\gamma$. Whereas when the variance of the multiplicative noise is larger than that of the additive noise, the two bounds are proportional to $\gamma^2$. Therefore, the performance of the estimator deteriorates dramatically with the increase of the multiplicative noise when the variance of equivalent noise is fixed.

If back to the unquantized problem $\mathbf{y}={{({\mathbf H}+{\mathbf E})}^{\rm T}}\mathbf{w}+\mathbf{n}$, one will draw a contrary conclusion. We define two unquantized problems which are the same with the above situation. Using the same notation and assuming the variance of the equivalent noise is equal, the result is contrary to (\ref{inequ}). According to (12) of \cite{Yujie}, one has
\begin{align}
{\rm {tr}}(\tilde{\mathbf J}^{-1})\geq{\rm {tr}}({\mathbf J}^{-1}).\notag
\end{align}
This result demonstrates that when the measurement is unquantized and the variance of the equivalent noise is fixed, noise coupled the parameter information can help us to estimate the parameter.

\subsection{The Case of Multiplicative Noise Fixed}

Now we discuss the case in which the variance of the multiplicative noise is fixed. Because $\mathbf w$ is  deterministic, $\sigma_e^2$ can be viewed as a variable instead of $\|{\mathbf w}\|_2^2\sigma_e^2$. We consider two extreme cases. One is that $\sigma_n^2$ is zero. In this case, we have known that the corresponding FIM is singular, thus there does not exist a finite unbiased estimator for $\mathbf w$. In the other case, when $\sigma_n^2$ tends to infinity, according to (\ref{lambda_ii}), one has
\begin{align}
\lim_{\sigma_n^2 \rightarrow \infty}\lambda_{ii}=\lim_{\sigma_n^2 \rightarrow \infty}\frac{2}{{\pi \sigma_z^2}}{{\rm e}^{-\frac{{{\left( \mathbf{h}_{i}^{\rm{T}}\mathbf{w} \right)}^{2}}}{{\sigma _z^2}}}}=0,\notag
\end{align}
which implies that the CRLB on the MSE tends to infinity as $\sigma_n^2$ gradually increases. Except for these two cases, the CRLB on the MSE is finite. These results show that the additive noise has two opposing effects. On the one hand, it provides variant thresholds for the sign measurement, which is beneficial to the estimation. On the other hand, the additive noise increases the variance of the estimation \cite{Onkar1}. Therefore, there may exist an optimal variance of the additive noise which balances these opposing effects and minimizes the CRLB. The above analysis will be substantiated by an example later.

\subsection{The Case of Additive Noise Fixed}

When $\sigma_e^2$ is zero, the CRLB on the MSE is finite. Whereas when $\sigma_e^2$ tends to infinity, according to (\ref{lambda_ii}), one has
\begin{align}
\lim_{\sigma_e^2 \rightarrow \infty}\lambda_{ii}=\lim_{\sigma_e^2 \rightarrow \infty}\frac{2}{{\pi \sigma_z^2}}{\rm e}^{-\frac{{\left({\mathbf h}_i^{\rm T}{\mathbf w} \right)}^2}{\sigma _z^2}}=0.\notag
\end{align}
Thus the FIM tends to singular and the CRLB on the MSE tends to infinity. Intuitively, one may expect that the optimal variance of the multiplicative noise is zero. However, we will show that this is indeed not always true. There may exist an optimal nonzero variance of the multiplicative noise, as we will show in the following example.

\subsection{An Example}

An example is now illustrated to verify our analysis on all cases. Consider a scalar parameter estimation problem and the mean of the sensing matrix is ${\mathbf H}=[1,1,\cdots,1]$. According to (\ref{mse_crlb}), the CRLB is
\begin{align}
{\rm CRLB}(w)=\frac{2\pi \sigma_z^2}{N}\left(1+\frac{\sigma_e^2w^2}{\sigma_n^2}\right)^2\Phi \left(-\frac{w}{\sigma_z}\right)\Phi \left(\frac{w}{\sigma_z}\right){\rm e}^{\frac{w^2}{\sigma_z^2}}.\label{CRLB_minimize}
\end{align}
We wish to minimize the CRLB (\ref{CRLB_minimize}) in three cases, respectively. It is obvious that the minimum CRLB is attained at $\sigma_e^2=0$ when $\sigma_z^2$ is fixed. When either $\sigma_n^2$ or $\sigma_e^2$ is fixed, it is difficult to exactly analyze (\ref{CRLB_minimize}). Fortunately, by using the Chernoff bound for the CDF \cite{Proakis}
\begin{align}\label{Cher}
 \Phi \left(-\frac{w}{\sigma_z}\right)\Phi \left(\frac{w}{\sigma_z}\right)\leq \frac{1}{4}{\rm e}^{-\frac{w^2}{2\sigma_z^2}},
\end{align}
one can find an upper bound for ${\rm CRLB}(w)$ by substituting (\ref{Cher}) in (\ref{CRLB_minimize})
\begin{align}
{\rm CRLB}(w)&\leq\frac{\pi \sigma_z^2}{2N}\left(1+\frac{\sigma_e^2w^2}{\sigma_n^2}\right)^2{\rm e}^{\frac{w^2}{2\sigma_z^2}}.\label{chernoff_minimize}
\end{align}
%\notag\\&=(\sigma_n^2+w_0^2\sigma_e^2)(1+\frac{\sigma_e^2w_0^2}{\sigma_n^2})^2{\textrm e}^{\frac{w_0^2}{2(\sigma_n^2+w_0^2\sigma_e^2)}}
In fact, the Chernoff bound is a very tight approximation for finding the optimal value of $\sigma_n^2$ or $\sigma_e^2$, which will be shown later. By substituting (\ref{Var_z}) in (\ref{chernoff_minimize}) and dropping out the constant coefficient items, we define the natural logarithm of the right hand side of (\ref{chernoff_minimize}) as
\begin{align}
f(\sigma_n^2,\sigma_e^2)\triangleq 3\log(\sigma_n^2+\sigma_e^2w^2)+\frac{w^2}{2(\sigma_n^2+\sigma_e^2w^2)}-2\log\sigma_n^2. \label{cher_approx}
\end{align}
When $\sigma_e^2$ is fixed, we minimize (\ref{cher_approx}) with respect to $\sigma_n^2$. The optimal variance of the additive noise is approximated by
\begin{align}\label{sigman_opt}
{}_{\rm opt}{\sigma}_n^2 \approx {}_{\rm app}{\sigma}_n^2= \frac{w^2}{2}\left(\sqrt{9\sigma_e^4+\sigma_e^2+\frac{1}{4}}+\frac{1}{2}+\sigma_e^2\right).
\end{align}
This means that there exists an optimal additive noise that matches the multiplicative noise and the unknown parameter.
Whereas when the variance of the additive noise $\sigma_n^2$ is fixed, the optimal ${}_{\rm opt}{\sigma}_e^2$ is
\begin{align}\label{sigmae_opt}
    {}_{\rm opt}{\sigma}_e^2\approx {}_{\rm app}{\sigma}_e^2=
    \begin{cases}
       \frac{1}{6}-\frac{\sigma_n^2}{w^2}, &{\rm if}\; \frac{\sigma_n^2}{w^2}\leq \frac{1}{6};\\
       0,&{\rm otherwise}.
   \end{cases}
\end{align}
It seems unreasonable that the multiplicative noise may improve the performance of the estimation. By carefully studying the condition of (\ref{sigmae_opt}), one can find that the variance of the additive noise $\sigma_n^2$ should be very small compared to $w^2$. In this setting, the randomness introduced by the additive noise is so weak that suitable perturbation may improve the MSE performance. However, to estimate the parameter accurately, a very large number of measurements is needed to ensure enough the fluctuation of the measurements. Thus the ML estimator achieves the CRLB only when the number of measurements is very large. When the variance of additive noise $\sigma_n^2$ is comparable with the energy of parameter $w$, the randomness introduced by the additive noise suffices and ${}_{\rm opt}{\sigma}_e^2$ is zero.

Notice that the above analysis is established for a given $w$. Although the parameter $w$ is unknown in practice, the theoretical analysis is still useful in three aspects. First, it gives us an insight into the relationship between the additive noise and the perturbation. Second, the theoretical MSE performance limits for unbiased estimators is provided by choosing the optimal ${}_{\rm opt}{\sigma}_n^2$ or ${}_{\rm opt}{\sigma}_e^2$. Third, one may extend the above ideas to the case of unknown parameter $w$. In this case, one may optimize the Bayesian CRLB \cite{Balkan} or the worst CRLB \cite{Oppen} instead if some prior information is known.

\section{ML Estimation via Convex Optimization}

At first sight, one wish the ML estimation problem (\ref{gen_p}) could be solved by steepest descent or Newton's method. However, problem (\ref{gen_p}) is non-convex. Hence the gradient and Hessian based numerical algorithms may not be guaranteed to converge to the optimal point. Moreover, direct solution can not provide us more insight into the problem itself. Fortunately, (\ref{gen_p}) can be reformulated as a convex optimization problem. By introducing a new variable
\begin{align}
{\mathbf v} = \frac{\mathbf w}{\sqrt{\|\mathbf w\|_2^2\sigma_e^2+\sigma_n^2}}, \label{tansform}
\end{align}
we transform the original optimization problem (\ref{gen_p}) to another one with respect to $\mathbf v$.
According to (\ref{tansform}), one has
\begin{align}\label{cons_norm_w}
\|{\mathbf v}\|_2< \frac{1}{\sigma_e}.
\end{align}
As long as ${\mathbf v}$ satisfies the inequality (\ref{cons_norm_w}), the relationship between $\mathbf w$ and ${\mathbf v}$ is a one to one mapping. Consequently, the original problem (\ref{gen_p}) can be conquered  by first solving an equivalent convex optimization problem,
\begin{subequations}\label{OPT}
\begin{align}
&\underset{\mathbf v}{\operatorname{minimize}}~-\sum_{i=1}^N{\log}\Phi({y_i{\mathbf h}_{i}^{\rm{T}}{\mathbf v}})\label{OBJ_OPT}\\
&{\operatorname{subject~to}}~\|{\mathbf v}\|_2^2<\frac{1}{\sigma_e^2},\label{Con_OPT}
\end{align}
\end{subequations}
and then finding the optimal point by
\begin{align}\label{invtrans}
{\mathbf w}=\frac{\sigma_n}{{\sqrt{1-\sigma_e^2\|{{\mathbf v}}\|_2^2}}}{\mathbf v}.
\end{align}
\begin{proposition}\label{opt_equva}
The optimization problem (\ref{gen_p}) is equivalent to problem (\ref{OPT}).
\end{proposition}
\begin{proof}
The proof is direct and is not included.
\end{proof}
%Let $r(\cdot)$ and $g(\cdot)$ denote the objective functions of problem (\ref{gen_p}) and (\ref{OPT}), and the corresponding optimal points are ${\mathbf w}^*$ and ${\mathbf v}^*$, respectively. On the one hand, ${\mathbf w}^*$ is the optimal point of (\ref{gen_p}), we can construct a unique ${\mathbf v}$ which satisfies the constraint (\ref{Con_OPT}). We have $r({\mathbf w}^*) = g({\mathbf v})\geq g({\mathbf v}^*)$. On the other hand, we can also show $g({\mathbf v}^*)\geq r({\mathbf w}^*)$ similarly. Therefore, $g({\mathbf v}^*)= r({\mathbf w}^*)$, the optimization problems are equivalent.
%Note that Proposition \ref{opt_equva} states that the optimal values are equal. Once we know an optimal point of either problem, the optimal point of the other problem can be determined. It is important to emphasize that Proposition \ref{opt_equva} does not demonstrate that the original ML estimator in (\ref{gen_p}) is convex in $\mathbf w$. It is the change in variables that makes the transformed problem convex in $\mathbf v$.Since the objective function of the problem (\ref{OPT}) is convex, the constraint can be directly discarded.

The constraint set (\ref{Con_OPT}) is an open ball. Thus if the optimal point of the problem (\ref{OPT}) exists, it must be an interior point of the constraint set. For the uniqueness of the optimal point of the objective function (\ref{OBJ_OPT}), we have the following result.
\begin{proposition}\label{uniq}
The objective function of problem (\ref{OPT}) is strictly convex, thus the optimal point of problem (\ref{OPT}) is unique if it exists.
\end{proposition}
\begin{proof}
The proof is postponed to Appendix D.
\end{proof}

Based on Proposition \ref{opt_equva} and Proposition \ref{uniq}, the following proposition provides a necessary and sufficient condition of the existence of the optimal point of the original ML estimation problem (\ref{gen_p}).

\begin{proposition}\label{equ}
The optimal point of problem (\ref{gen_p}) exists if and only if the optimal point ${\mathbf v}_{u}^*$ of the unconstrained convex optimization problem
\begin{align}
&\underset{{\mathbf v}\in \mathbb{R}^p}{\operatorname{minimize}}~-\sum_{i=1}^N{\log}\Phi({y_i{\mathbf h}_{i}^{\rm{T}}{\mathbf v}})\label{UNOPT}
\end{align}
satisfies the constraint (\ref{Con_OPT}).
\end{proposition}

\begin{proof}
According to Proposition \ref{opt_equva}, the optimal point of the original ML estimation problem (\ref{gen_p}) exists if and only if problem (\ref{OPT}) has an optimal point.
Considering that the objective function of (\ref{OPT}) is strictly convex and the constraint (\ref{Con_OPT}) is an open ball, the existence of optimal point of (\ref{OPT}) is equivalent to that ${\mathbf v}_{u}^*$ satisfies the constraint (\ref{Con_OPT}). Then the result is established.
\end{proof}

Therefore, we can solve at first the unconstrained optimization problem (\ref{UNOPT}). Then we check whether ${\mathbf v}_{u}^*$ satisfies the constraint (\ref{Con_OPT}) to determine whether the original ML estimation problem (\ref{gen_p}) has an optimal point.

It is shown that (\ref{gen_p}) does not have an optimal point in some cases, in which we say that there ``exists'' an optimal point with infinite norm. In order to provide an finite estimation, we adopt a norm limit operation.
We will project the optimal point (the infinite case included) onto a set ${\mathcal W}=\{{\mathbf w}|\|{\mathbf w}\|_2\leq R_w\}$, if its norm is larger than a threshold $R_w$, where $R_w$ is much larger than the norm of true parameter ${\bf w}_0$. Because ${\mathcal W}$ is a closed ball, it can be proved that the projection of ${\bf w}$ onto ${\mathcal W}$ is equivalent to the projection of ${\bf v}$ onto ${\mathcal V}= \{{\mathbf v}|\|\mathbf v\|_2\leq {R_w}/{{\sqrt{R_w^2\sigma_e^2+\sigma_n^2}}}\}$.

In practice, we first get ${\mathbf v}_{p}^*=\Pi_{\mathcal V}({\mathbf v}_{u}^*)$, where $\Pi_{\mathcal V}$ represents the projection onto the parameter set ${\mathcal V}$. According to (\ref{invtrans}), we then obtain the corresponding ${\mathbf w}_{p}^*$ which satisfies ${\mathbf w}_{p}^*\in {\mathcal W}$ and is regarded as the ML estimator
\begin{align}\label{ML}
{\hat{\mathbf w}}_{\rm ML}=\frac{\sigma_n}{{\sqrt{1-\sigma_e^2\|{{\mathbf v}_{p}^*}\|_2^2}}}{\mathbf v}_{p}^*.
\end{align}

We will show that the norm limit operation is almost unnecessary, when the number of measurements $N$ is large enough. As shown in Theorem \ref{ML_con}, the ML estimator is consistent. This means that the optimal point of (\ref{gen_p}) converges to ${\mathbf w}_0$ in probability. Thus the optimal point of (\ref{OPT}) also converges to ${\mathbf v}_0$ in probability, where ${\mathbf v}_0$ is determined by (\ref{tansform}) using ${\mathbf w}_0$. Because ${\mathbf w}_0\in{\mathcal W}$, we can see that ${\mathbf v}_0\in{\mathcal V}$. As a consequence, in the situation of large measurement set, the optimal point ${\mathbf v}_u^*$ of (\ref{UNOPT}) satisfies ${\mathbf v}_u^* \in {\mathcal V}$ with high probability.

\section{Further Discussion}

\subsection{Probability Analysis}

We will analyze the probability that the optimal point of the ML estimation problem exists, in which situation the likelihood function is unimodal.

The previous section has shown that the optimal point ${\mathbf v}_u^*$ of (\ref{UNOPT}) may violate the constraint (\ref{Con_OPT}). Since ${\mathbf v}_u^*$ is a random vector, we may define the probability
\begin{align}
{\rm{P}_{\mathcal V}}={\rm{Pr}}\left[\|{\mathbf v}_u^*\|_2< \frac{1}{\sigma_e}\right],\label{Pr_V}
\end{align}
which is the probability that the original log-likelihood function (\ref{Loglike}) is unimodal. This probability is meaningful in that it sheds light upon the perturbation in a different perspective. On the one hand, the perturbation provides randomness for the sign function. The randomness may make the likelihood function unimodal. On the other hand, ${\rm P}_{\mathcal V}$ may decrease as the strength of perturbation becomes larger through (\ref{Pr_V}). Notice that this probability can be written explicitly as ${\rm P}_{\mathcal V}({\mathbf H},{\mathbf w}_0,\sigma_e^2,\sigma_n^2,N)$, where ${\mathbf w}_0$ denotes the true value of $\mathbf w$. Computing the above probability seems to be hard. However, there may exist an analytic solution in a special case.
\begin{proposition}\label{Pr_b}
Suppose that the true parameter $w_0$ is a scalar, and the mean of the sensing matrix is ${\mathbf H}=[1,1,\cdots,1]$. By defining $k^-=\left\lfloor N\Phi\left(-\frac{\textstyle 1}{\textstyle \sigma_e}\right)\right\rfloor+1$ and ${k^+}=\left\lceil N\Phi\left(\frac{\textstyle 1}{\textstyle \sigma_e}\right)\right\rceil-1$, ${\rm P}_{\mathcal V}$ is obtained by
\begin{align}\label{Cal_pro}
{\rm P}_{\mathcal V}=\sum_{k={k^-}}^{k^+}{N \choose k}\left[\Phi\left(\frac{w_0}{\sigma_z}\right)\right]^k\left[1-\Phi\left(\frac{w_0}{\sigma_z}\right)\right]^{N-k}.
\end{align}
\end{proposition}
\begin{proof}
The proof is postponed to Appendix E.
\end{proof}
%We may have some remarks on this simple case, and some conclusion could be extended to the general case.
%\begin{remark}
%Since ${\textstyle k_{min}}$ and ${\textstyle k_{max}}$ introduce the floor and ceil operation, respectively, analysis of the probability ${\textrm{P}_{\mathcal V}}$ seems very hard.

It is difficult to analyze ${\rm{P}_{\mathcal V}}$ by (\ref{Cal_pro}). When the number of measurements $N$ is large, however, we may compute ${\rm{P}_{\mathcal V}}$ with the normal approximation \cite{Papoulis}
\begin{align}
{\rm{P}_{\mathcal V}}\approx \Phi(\eta^{+})-\Phi(\eta^{-}),\notag
\end{align}
where
\begin{align}
\eta^{\pm}=\frac{k^{\pm}-N\Phi\left(\frac{w_0}{\sigma_z}\right)}{\sqrt{N\Phi\left(\frac{w}{\sigma_z}\right)\Phi\left(-\frac{w}{\sigma_z}\right)}}\approx\frac{\Phi\left(\pm\frac{1}{\sigma_e}\right)-\Phi\left(\frac{w_0}{\sigma_z}\right)}{\sqrt{\Phi\left(\frac{w_0}{\sigma_z}\right)\Phi\left(-\frac{w_0}{\sigma_z}\right)}}\sqrt{N}.\notag
\end{align}
%\eta_2\approx \frac{\Phi\left(\frac{1}{\sigma_e}\right)-\Phi\left(\frac{w_0}{\sigma_z}\right)}{\sqrt{\Phi\left(\frac{w_0}{\sigma_z}\right)\Phi\left(-\frac{w_0}{\sigma_z}\right)}}\sqrt{N}.\notag
%\end{align}
%=\frac{{\textstyle k_{\rm min}}-N\Phi\left(\frac{w_0}{\sigma_z}\right)}{\sqrt{N\Phi\left(\frac{w_0}{\sigma_z}\right)\Phi\left(-\frac{w_0}{\sigma_z}\right)}}
%=\frac{{\textstyle k_{\rm max}}-N\Phi\left(\frac{w_0}{\sigma_z}\right)}{\sqrt{N\Phi\left(\frac{w_0}{\sigma_z}\right)\Phi\left(-\frac{w_0}{\sigma_z}\right)}}
%\begin{subequations}
%\begin{align}
%t_1=\frac{{\textstyle k_{min}}-N\Phi\left(\frac{w_0}{\sigma_z}\right)}{\sqrt{N\Phi\left(\frac{w_0}{\sigma_z}\right)\Phi\left(-\frac{w_0}{\sigma_z}\right)}}\approx \frac{\Phi\left(-R_v\right)-\Phi\left(\frac{w_0}{\sigma_z}\right)}{\sqrt{\Phi\left(\frac{w_0}{\sigma_z}\right)\Phi\left(-\frac{w_0}{\sigma_z}\right)}}\sqrt{N},\\
%t_2=\frac{{\textstyle k_{max}}-N\Phi\left(\frac{w_0}{\sigma_z}\right)}{\sqrt{N\Phi\left(\frac{w_0}{\sigma_z}\right)\Phi\left(-\frac{w_0}{\sigma_z}\right)}}\approx \frac{\Phi\left(R_v\right)-\Phi\left(\frac{w_0}{\sigma_z}\right)}{\sqrt{\Phi\left(\frac{w_0}{\sigma_z}\right)\Phi\left(-\frac{w_0}{\sigma_z}\right)}}\sqrt{N}.
%\end{align}
%\end{subequations}

Now we give some results using the normal approximation. Note that the conclusion only applies to the case of large enough measurements.
The normal approximation indicates that ${\rm P}_{\mathcal V}$ is an increasing function of $N$ and $\underset{N\rightarrow\infty}{\lim}{{\rm P}_{\mathcal V}} = 1$. In fact, ${\rm P}_{\mathcal V}$ is not a monotone increasing function of $N$, but the overall trend of ${\rm P}_{\mathcal V}$ is increasing in $N$. When $\sigma_z$ and $N$ are fixed, the higher value of $\sigma_e$, the smaller value of ${{\rm P}_{\mathcal V}}$. This means that for a larger $\sigma_e$, more observations is needed to make sure that ${\rm P}_{\mathcal V}$ achieves a given probability.
%When either $\sigma_n$ or $\sigma_e$ is fixed, it is hard to analyze ${\rm P}_{\mathcal V}$.

In view of the limited number of observations, ${\rm P}_{\mathcal V}$ is always less than $1$. Therefore, for any fixed $N$, the probability that the estimated parameter lies in the boundary of the parameter set ${\mathcal W}=\{{\mathbf w}|\|{\mathbf w}\|_2\leq R_w\}$ is nonzero. Since $R_w$ is larger than the norm of ${\mathbf w}_0$, the estimator has significant bias in the case that $\|\hat{w}_{\rm ML}\|_2=R_w$.
% Given all the parameters, we can determine the probability that the log-likelihood function is unimodal.
%The theoretical analysis of ${\rm P}_{\mathcal V}$ provides insight into the tradeoff between the randomness introduced by perturbation and its intrinsic constraint.
% When the number of measurements is not large enough, the optimal point of the problem (\ref{gen_p}) may not exist without the constraint. Therefore, the parameter set $\mathcal W$ is essential in estimating the parameter only when the number of measurements is small.
%
%It is worth noting that the constrained ML estimation problem can also be formulated as a constrained convex optimization problem given by
%\begin{align}
%% \nonumber to remove numbering (before each equation)
%&\underset{{\mathbf v}\in {\mathcal{V}}}{\operatorname{minimize}}~-\left(\sum_{i=1}^N{\log}\Phi({y_i{\mathbf h}_{i}^{\rm{T}}{\mathbf v}})\right).\label{Con_ML}
%\end{align}
%We do not resort to solving this problem (\ref{Con_ML}). On the one hand, when the number of measurements is large enough, ${\mathbf v}_{u}^*$ will be automatically lying in the parameter set $\mathcal V$. On the other hand, when ${\mathbf v}_{u}^*$ violates the constraint set $\mathcal V$, the optimal point of problem (\ref{Con_ML}) must be lying in the boundary of the parameter set. Hence, the solution of our method is approximately well to the optimal point of (\ref{Con_ML}).

\subsection{Effects of Mismodeling on the ML Estimation}

In this subsection, we study the effects on estimation due to the ignorance of perturbation. Then we analyze the performance of both estimators.

Assume that the true data generating model is (\ref{Gaussian_model}), and the corresponding ML estimation is ${\mathbf w}_{t}^*$. If the perturbation is ignored, we obtain the corresponding ML estimation as ${\mathbf w}_{w}^*$ using model (\ref{Determin_matrix}). ${\mathbf w}_{w}^*$ is denoted as the perturbation-ignored estimator. Assume that ${\mathbf w}_{t}^* \in {\mathcal W}$. The relationship between ${\mathbf w}_{t}^*$ and ${\mathbf w}_{w}^*$ is given in the following proposition.
\begin{proposition}
The direction of ${\mathbf w}_{w}^*$ is the same with ${\mathbf w}_{t}^*$, with magnitude scaled. ${\mathbf w}_{w}^*$ and ${\mathbf w}_{t}^*$ satisfy the following relationship,
\begin{align}
{\mathbf w}_{w}^*= \frac{{{\mathbf w}_{t}^*}}{{\sqrt{1+\frac{\sigma_e^2}{\sigma_n^2}\|{\mathbf w}_{t}^*\|_2^2}}}.\label{Relation}
\end{align}
\end{proposition}
\begin{proof}
Substituting $\sigma_e^2=0$ in (\ref{invtrans}), one has
\begin{align}
{\mathbf w}_{w}^* &= \sigma_n {\mathbf v}_{u}^*. \label{opt_a}
\end{align}
Considering the existence of the perturbation, one has
\begin{align}
{\mathbf v}_{u}^* &= \frac{{\mathbf w}_t^*}{{\sqrt{\|{\mathbf w}_t^*\|_2^2\sigma_e^2+\sigma_n^2}}}.\label{opt_b}
\end{align}
Using (\ref{opt_a}) and (\ref{opt_b}) to eliminate ${\mathbf v}_{u}^*$, the equation (\ref{Relation}) is established.
\end{proof}

In some applications such as binary regression problems, the direction of the parameter $\mathbf w$ is much more important than the magnitude of $\mathbf w$. In this situation, even though we do not know the strength of perturbation, we can still estimate the direction of $\mathbf w$ by perturbation-ignored estimator using model (\ref{Determin_matrix}). Similar result in binary regression problem is obtained in \cite{Burrdoss}, which focuses on model (\ref{Binary_reg}) with sensing matrix perturbation and is consistent with Proposition 5.

According to the analysis above, if the information of  the perturbation ${\bf E}$ and additive noise ${\bf n}$ is both unknown, and the measurements are generated by model (\ref{Gaussian_model}), we can still use model (\ref{Binary_reg}) to estimate the direction of the parameter ${\bf w}$.

Now we discuss the performance of both estimators. Both the parameters ${\mathbf w}_0$ and $\sigma_n^2$ are supposed to be fixed. We show that each estimator has its own advantages and disadvantages, and the performance of the estimator depends on the strength of perturbation $\sigma_e^2$ and the number of measurements $N$.

When the number of measurements $N$ tends to infinity, ${\mathbf w}_{t}^*$ converges to ${\mathbf w}_0$ in probability. According to (\ref{Relation}), the estimator ignored perturbation is inconsistent. For the vector parameter estimation problem, the square error between ${\mathbf w}_{w}^*$ and ${\mathbf w}_{t}^*$ is
\begin{align}\label{Mse_conv}
\|{\mathbf w}_{w}^*-{\mathbf w}_{t}^*\|_2^2=\|{\mathbf w}_{t}^*\|_2^2\left(1-\left(\frac{\sigma_e^2}{\sigma_n^2}\|{\mathbf w}_{t}^*\|_2^2+1\right)^{-\frac{1}{2}}\right)^2.\nonumber
\end{align}
According to the continuous mapping theorem \cite{larry}, the squared error converges to
\begin{align}\label{MSE_converge}
\|{\mathbf w}_{w}^*-{\mathbf w}_{t}^*\|_2^2 \stackrel{\rm p}{\rightarrow} \|{\mathbf w}_{0}\|_2^2\left(1-\left({\frac{\sigma_e^2}{\sigma_n^2}\|{\mathbf w}_{0}\|_2^2+1}\right)^{-\frac{1}{2}}\right)^2.
\end{align}
If ${\sigma_e^2}\ll {\sigma_n^2}/\|{\mathbf w}_{0}\|_2^2$, the squared error is approximately $\frac{\textstyle \sigma_e^4}{\textstyle 4\sigma_n^4}\|{\mathbf w}_0\|_2^6$, and the relative error is $\frac{\textstyle \sigma_e^2}{\textstyle 2\sigma_n^2}\|{\mathbf w}_0\|_2^2$.

We have discussed the probability that the ML estimator ${\mathbf w}_{t}^*$ is finite. When the number of measurements is not large enough, the ML estimator ${\mathbf w}_t^*$ has a much larger probability of being infinite than ${\mathbf w}_w^*$. Meanwhile, even if the ML estimator ${\mathbf w}_t^*$ is finite, the norm of ${\mathbf w}_{t}^*$ may be much larger than that of ${\mathbf w}_0$. In this situation, ${\mathbf w}_{t}^*$ may have a larger MSE than ${\mathbf w}_{w}^*$ despite the bias of ${\mathbf w}_{w}^*$.

To sum up, the direction of perturbation-ignored estimator is the same with that of ML estimator.
If the perturbation and additive noise is both unknown, we can still estimate the direction of the unknown parameter vector.
Although the perturbation-ignored estimator is inconsistent and biased, it works better in the MSE sense when the number of measurements is not large.

\section{Numerical Simulations}

In this section, several numerical simulations are performed to verify the theoretical results presented in previous sections. In these simulations, when the unknown parameter $w$ is a scalar, the mean sensing matrix is chosen as ${\mathbf H}=[1,1,\cdots,1]$. Whereas when the unknown parameter $\mathbf w$ is a vector, the mean sensing matrix $\mathbf H$ is drawn with each entry $h_{ij}\sim {\mathcal N}(0,1)$. All the MSEs of the ML estimator are averaged over 2000 Monte Carlo (MC) trials unless stated otherwise. We assume that $R_w=4\|{{\mathbf w}_0}\|_2$. The binary measurements are generated by model (\ref{Gaussian_model}).
\subsection{Validation of the Performance Limits}
The first simulation is to validate the correctness of Proposition \ref{Pro_ineq}. The data is generated as follows. We set $p = 4$, $N=300$, $\sigma_z^2=4\|{\mathbf w}_0\|_2^2$, and generate the true parameter ${\mathbf w}_0\in {\mathbb{R}}^p$ from ${\mathcal N}({\mathbf 0},{\mathbf I})$. The results are plotted in Fig. \ref{fig_ine} with $\gamma$ varying from $10^{-2}$ to $10^2$. It can be seen that the bounds are proportional to $\gamma$ when $\gamma$ is small, while they are proportional to $\gamma^2$ when $\gamma$ is large. 
%It can be seen that the performance worsens dramatically with the increase of the variance of the perturbation noise when $\sigma_z^2$ is fixed.
\begin{figure}[h!t]
\centering
\includegraphics[width=3.5in]{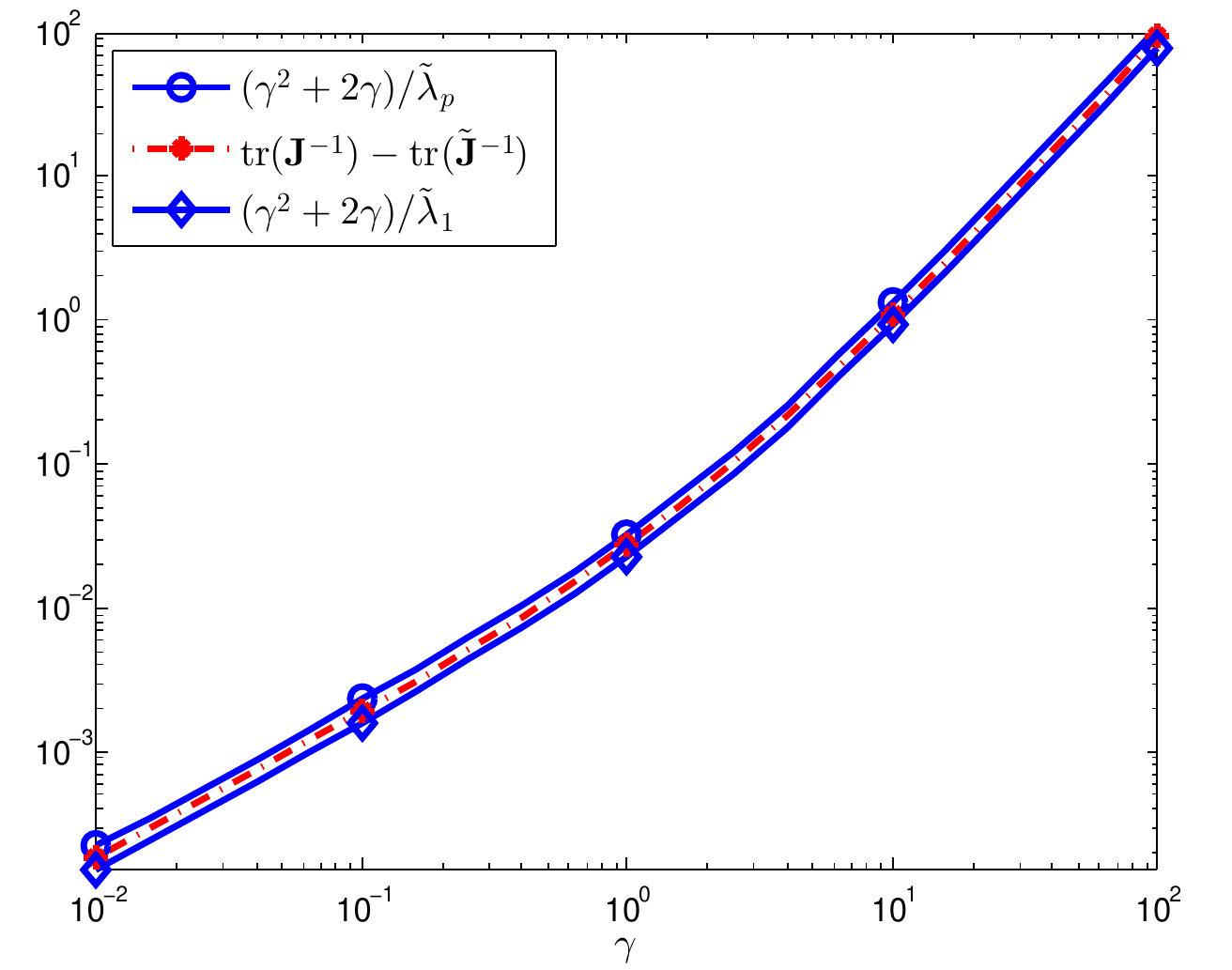}
% where an .eps filename suffix will be assumed under latex,
% and a .pdf suffix will be assumed for pdflatex; or what has been declared
% via \DeclareGraphicsExtensions.
\caption{Validation of Proposition \ref{Pro_ineq} with a log-log plot. All the three curves are approximately piecewise.}
\label{fig_ine}
\end{figure}

The second simulation is to validate the existence of the  optimal variance of additive noise when $\sigma_e^2$ is fixed. The parameters are selected as $\sigma_e^2=0.3$ and $w_0=1$. It can be seen that the Chernoff bound is a very accurate approximation of the CRLB. Meanwhile, it is also demonstrated that (\ref{sigman_opt}) is a good approximation of the optimal variance of the additive noise.
\begin{figure}[h!t]
\centering
\includegraphics[width=3.5in]{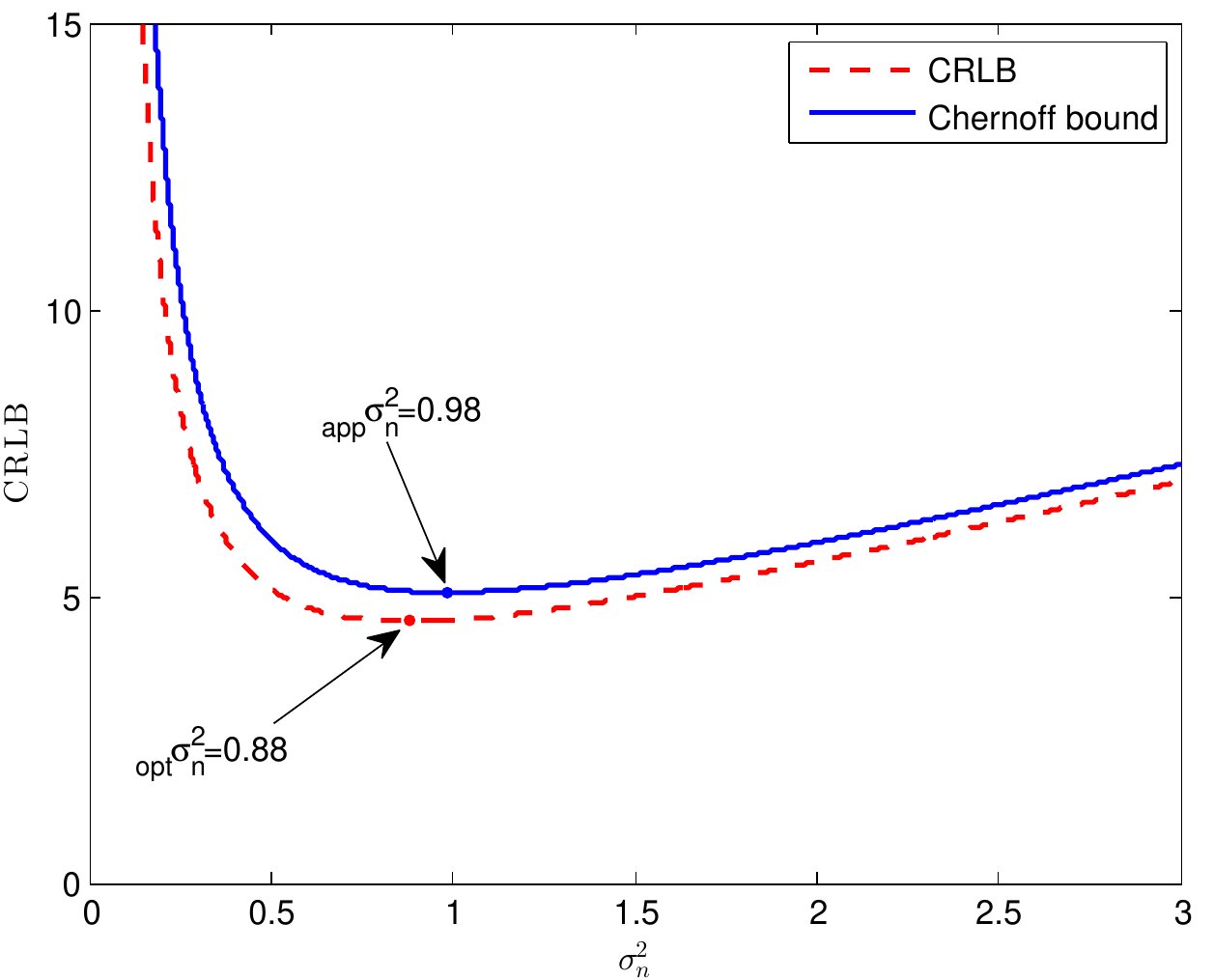}
% where an .eps filename suffix will be assumed under latex,
% and a .pdf suffix will be assumed for pdflatex; or what has been declared
% via \DeclareGraphicsExtensions.
\caption{The optimal variance of the additive noise when $\sigma_e^2$ is fixed as $0.3$. Note that the red dashed line indicates the CRLB, while the blue solid line indicates the Chernoff bound. The red point denotes the optimal value ${}_{\rm opt}\sigma_n^2=0.88$, and the blue point denotes the approximated value ${}_{\rm app}\sigma_n^2=0.98$.}
\label{optimal_add_noise}
\end{figure}

In the third simulation, the existence of the optimal strength of perturbation is validated when $\sigma_n^2$ is fixed. We set $\sigma_n^2=0.1$ and $w_0=1$. In this case, $\sigma_n^2/w_0^2\leq 1/6$ and the optimal strength of perturbation exists. Although the Chernoff bound is not very tight when $\sigma_e^2$ is near $0$, it still reflects the trend of the CRLB well.
\begin{figure}[h!t]
\centering
\includegraphics[width=3.5in]{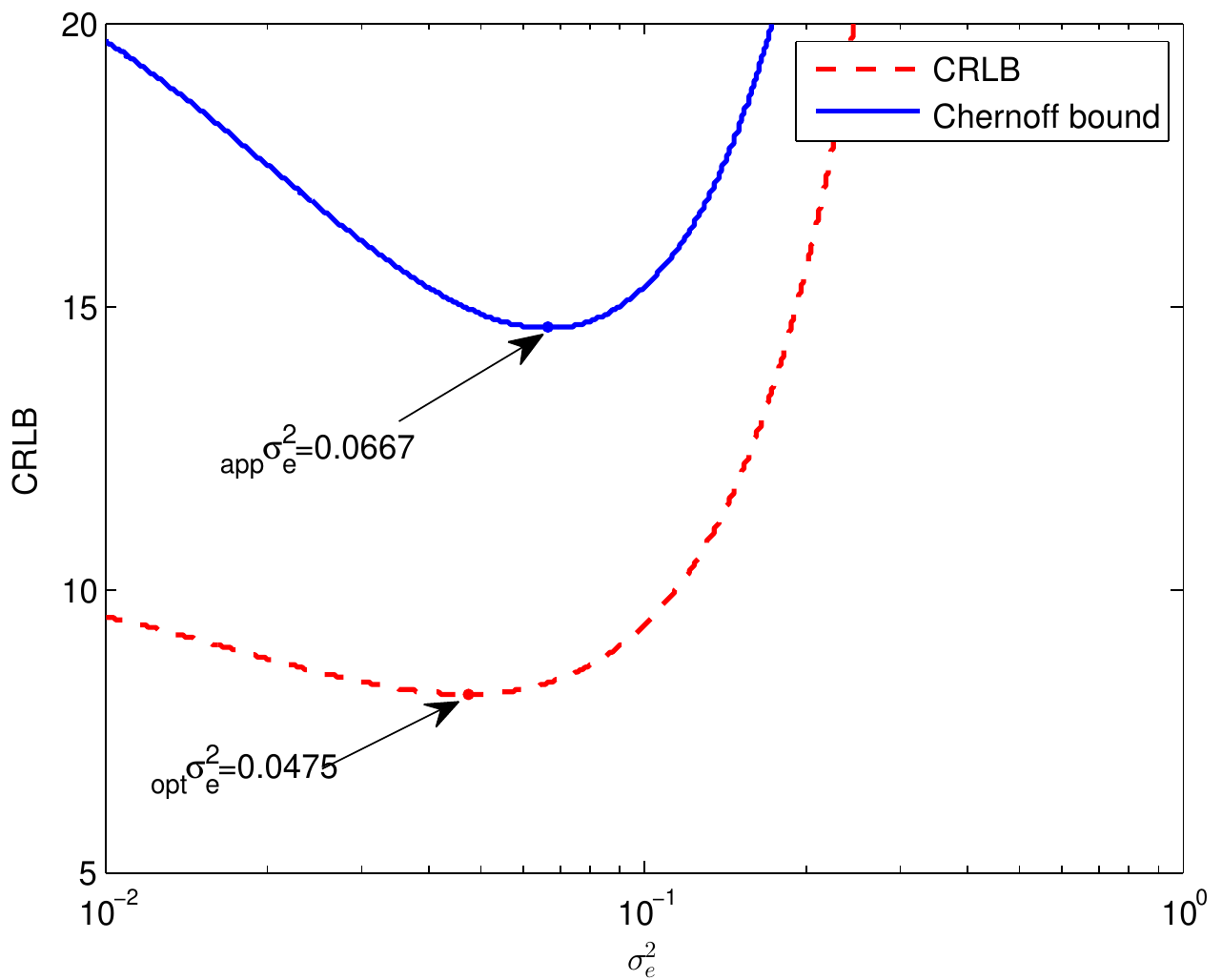}
% where an .eps filename suffix will be assumed under latex,
% and a .pdf suffix will be assumed for pdflatex; or what has been declared
% via \DeclareGraphicsExtensions.
\caption{The optimal strength of the perturbation when $\sigma_n^2$ is fixed as $0.1$. The red dashed line indicates the CRLB, while the blue solid line indicates the Chernoff bound. The red point denotes the optimal value ${}_{\rm opt}\sigma_e^2=0.0475$, and the blue point denotes the approximated value ${}_{\rm app}\sigma_e^2=0.0667$.}
\label{optimal_pertu_noise}
\end{figure}
%It can be seen that the upper bound and lower bound are both tight.
%$\gamma = {\frac{\textstyle \sigma_e^2}{\textstyle \sigma_n^2 }\|{\mathbf w}_0\|^2}$,
%The upper bound, the true difference of the MSE, and the lower bound versus the relative strength of the perturbation noise $\gamma$}
\subsection{Simulations of Probability Results}
The first simulation is to substantiate Proposition \ref{opt_equva}. The parameters are selected as: $\sigma_n^2=1$, $\sigma_e^2=0.5$, $w_0=1$ and $N=40$. Two typical realizations about the negative log-likelihood function in (\ref{gen_p}) versus $w$ are plotted in Fig. \ref{fig_sim}. It can be seen that the negative log-likelihood function (\ref{gen_p}) is nonconvex. It is also shown that the original problem (\ref{gen_p}) has an optimal point if $v_u^*$ satisfies the constraint (\ref{Con_OPT}). Whereas when $v_u^*$ violates the constraint (\ref{Con_OPT}), the optimal point of problem (\ref{gen_p}) does not exist.
\begin{figure}[h!t]
\centering
\includegraphics[width=3.5in]{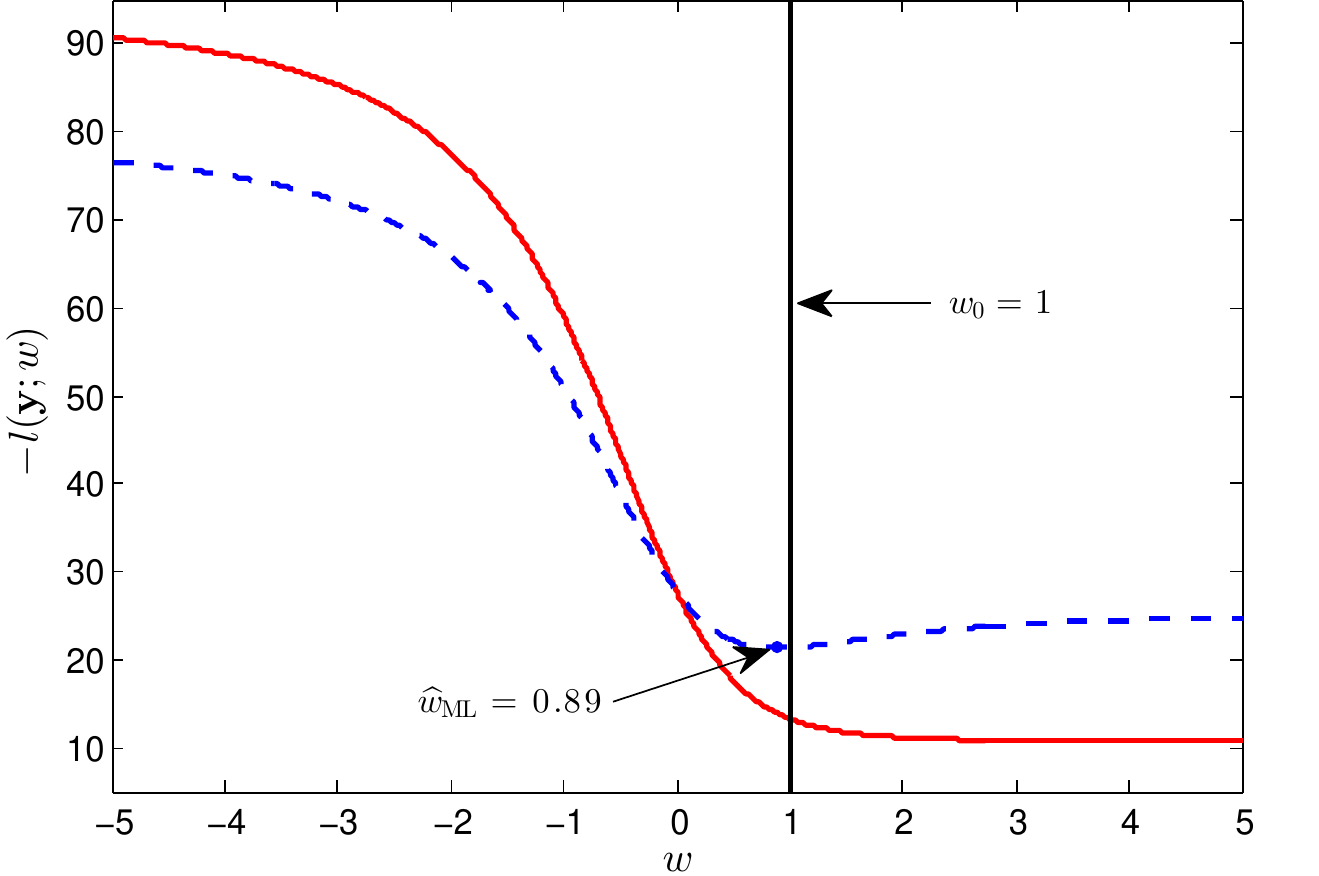}
% where an .eps filename suffix will be assumed under latex,
% and a .pdf suffix will be assumed for pdflatex; or what has been declared
% via \DeclareGraphicsExtensions.
\caption{The negative log-likelihood (\ref{gen_p}) versus the parameter $w$. The blue dashed line indicates that $v_u^*$ does not violate the constraint (\ref{Con_OPT}), while the red solid line indicates that $v_u^*$ violates the constraint (\ref{Con_OPT}). The black vertical solid line denotes the true value $w_0=1$, and the blue point denotes the estimated value $\hat{w}_{\rm ML}=0.89$, corresponding to the optimal point of the blue dashed line. }
\label{fig_sim}
\end{figure}

The next two simulations focus on the probability that the log-likelihood function is unimodal. The probability ${\rm P}_{\mathcal V}$ is computed by (\ref{Cal_pro}) and two cases are considered. For the first case, we assume that $\sigma_n^2=0.3$ and $w_0=1$. Note that the variance of the additive noise is small compared with $w_0$. From Fig. \ref{fig_pro1}, we see that ${\rm P}_{\mathcal V}$ is not monotonically increasing with $N$. This is mainly because of the floor and ceil operations. Nevertheless, the overall trend is increasing with $N$. Meanwhile, for a fixed $N$, it is shown that ${\rm P}_{\mathcal V}$ increases with $\sigma_e^2$ when $\sigma_e^2$ is small. When the strength of perturbation is large, ${\rm P}_{\mathcal V}$ decreases with it. Therefore, suitable perturbation can improve the probability ${\rm P}_{\mathcal V}$ when $N$ is not too large. For the second case, we set $\sigma_z^2=2$ and $w_0=1$. In Fig. \ref{fig_pro2}, it can be seen that the larger the strength of perturbation, the smaller the probability ${\rm P}_{\mathcal V}$ is. This result is consistent with our analysis by the normal approximation.
\begin{figure}[h!t]
\centering
\includegraphics[width=3.5in]{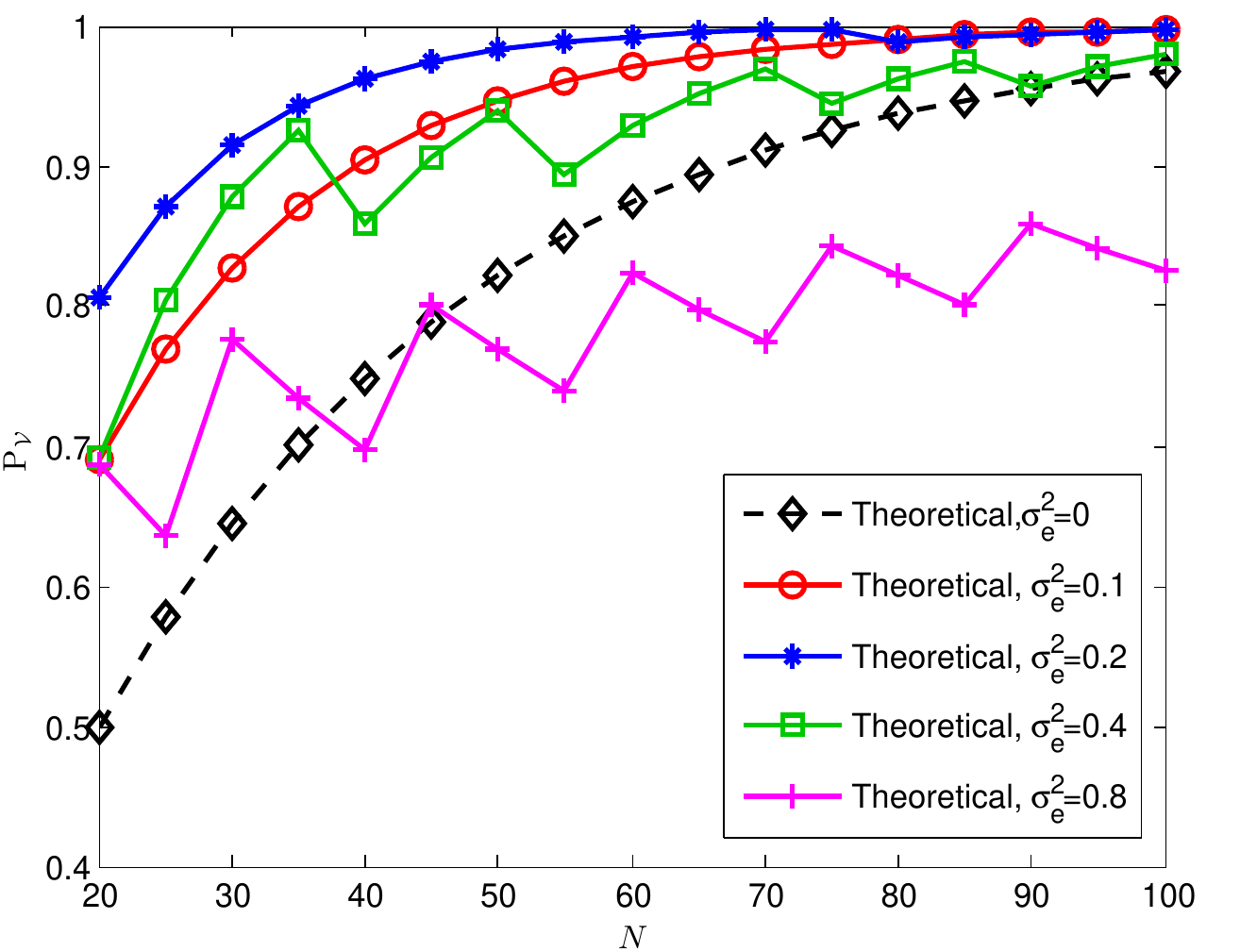}
\caption{The probability ${\textrm{P}_{\mathcal V}}$ versus the number of observations $N$ when $\sigma_n^2$ is fixed as $0.3$. }
\label{fig_pro1}
\end{figure}
\begin{figure}[h!t]
\centering
\includegraphics[width=3.5in]{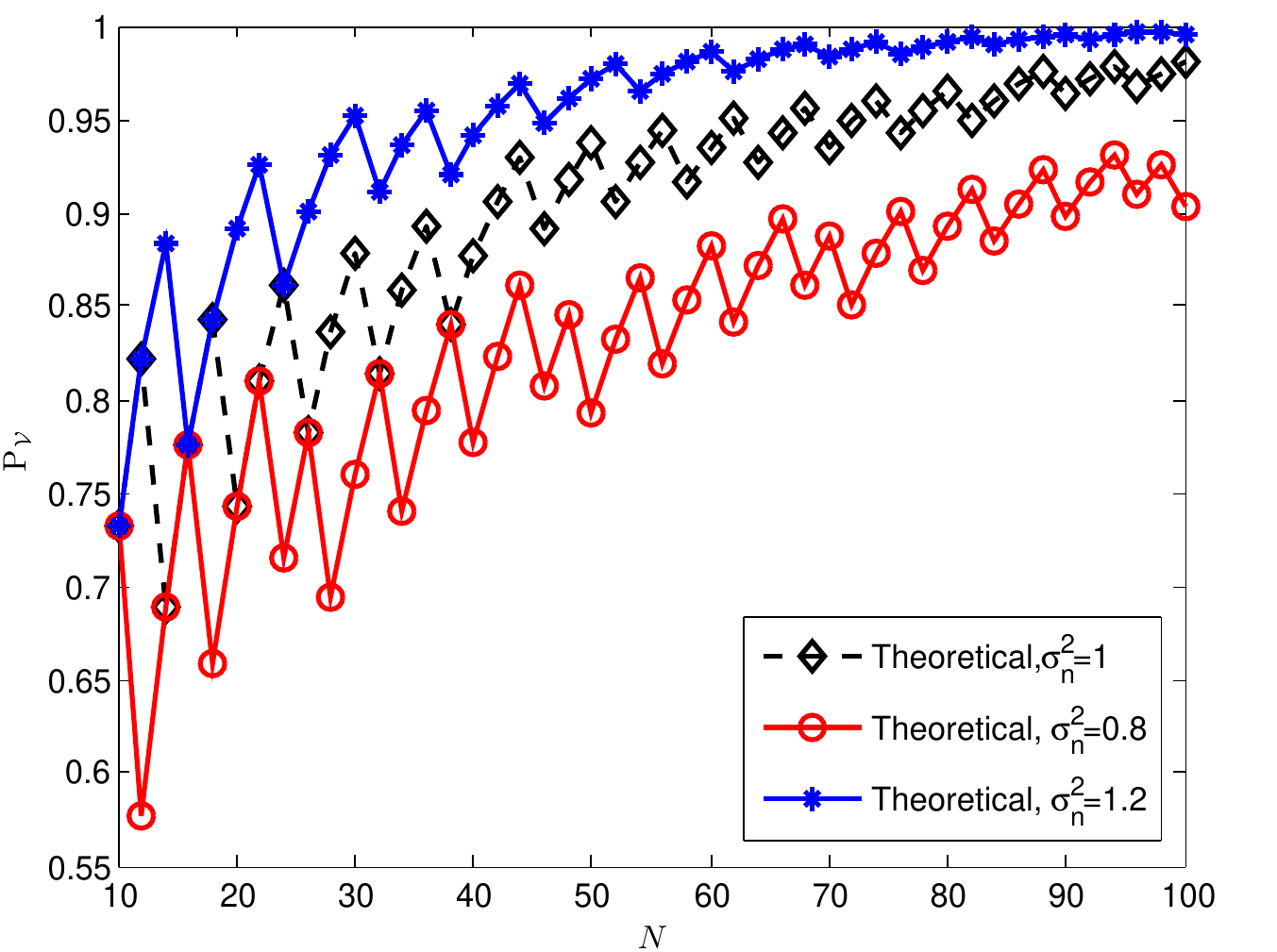}
\caption{The probability ${\textrm{P}_{\mathcal V}}$ versus the number of observations $N$ when $\sigma_z^2$ is fixed as $2$.}
\label{fig_pro2}
\end{figure}
\subsection{Performance of ML Estimator}
In this subsection, several MC simulations are performed to evaluate the MSE performance of the ML estimator against the CRLB. The MATLAB \emph{fminunc} function is used to solve the problem \cite{Geletu}. We assume that ${\mathbf w}_0=[0.7,0.5,-0.6]^{\rm T}$ in the first and the third simulations.

In the first simulation, we compare the MSE performance of the ML estimator with the CRLB with fixed $\sigma_z^2$ or $\sigma_n^2$. The results are plotted in Fig. \ref{fig_mse1}. It is shown that the performance worsens with the increase of the proportion of the multiplicative noise when $\sigma_z^2$ is fixed. When $\sigma_n^2$ is fixed, it is demonstrated that the perturbation exacerbates the performance of estimation. The reason is that the norm of ${\mathbf w}_0$ is comparable with $\sigma_n$ in this case, and the optimal ${}_{\rm opt}{\sigma_e^2}$ is zero. Note that the MSE of the ML estimator is much larger than the corresponding CRLB when $N$ is small. This is mainly because the norm of the ML estimator has a substantial probability of being $R_w$.
\begin{figure}[h!t]
\centering
\includegraphics[width=3.5in]{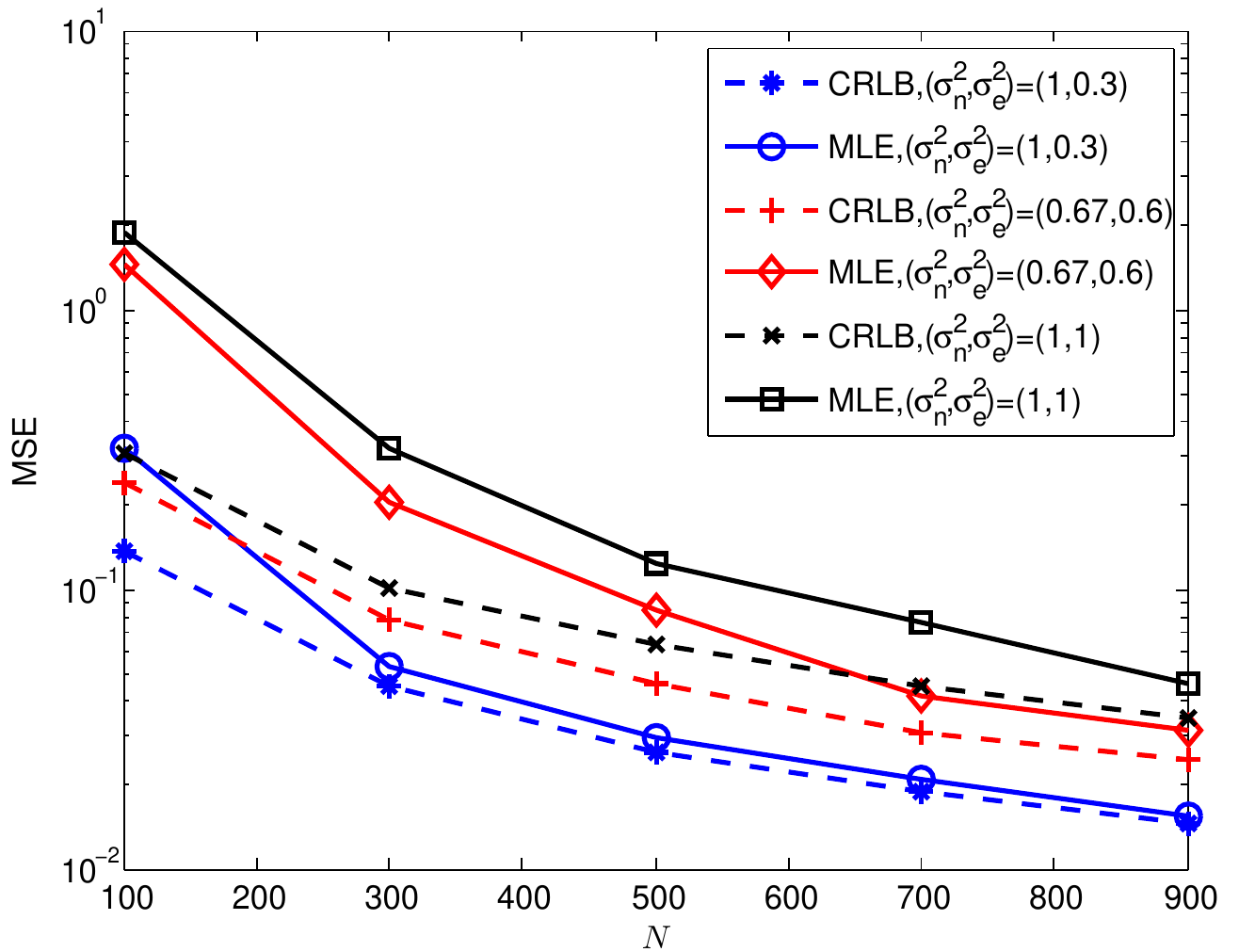}
\caption{The MSE of the ML estimator and the CRLB for different number of observations. Note that the blue and the red lines correspond to the case that the variance of equivalent noise $\sigma_z^2$ is fixed. While the blue and the black lines correspond to the case that the variance of additive noise $\sigma_n^2$ is fixed.}
\label{fig_mse1}
\end{figure}

The second simulation assumes that $\sigma_e^2=0.3$ and $w_0=1$, and the MSEs are averaged over $5000$ MC trials. The results are plotted in Fig. \ref{fig_mse3}. The optimal variance of the additive noise is $0.88$, thus its CRLB is smaller than the $\sigma_n^2=2$ case. Meanwhile, the ML estimator approaches faster to the CRLB in the $\sigma_n^2=2$ case. As the number of measurements increases, the ML estimator also attains the CRLB in the $\sigma_n^2=0.88$ case.
\begin{figure}[h!t]
\centering
\includegraphics[width=3.5in]{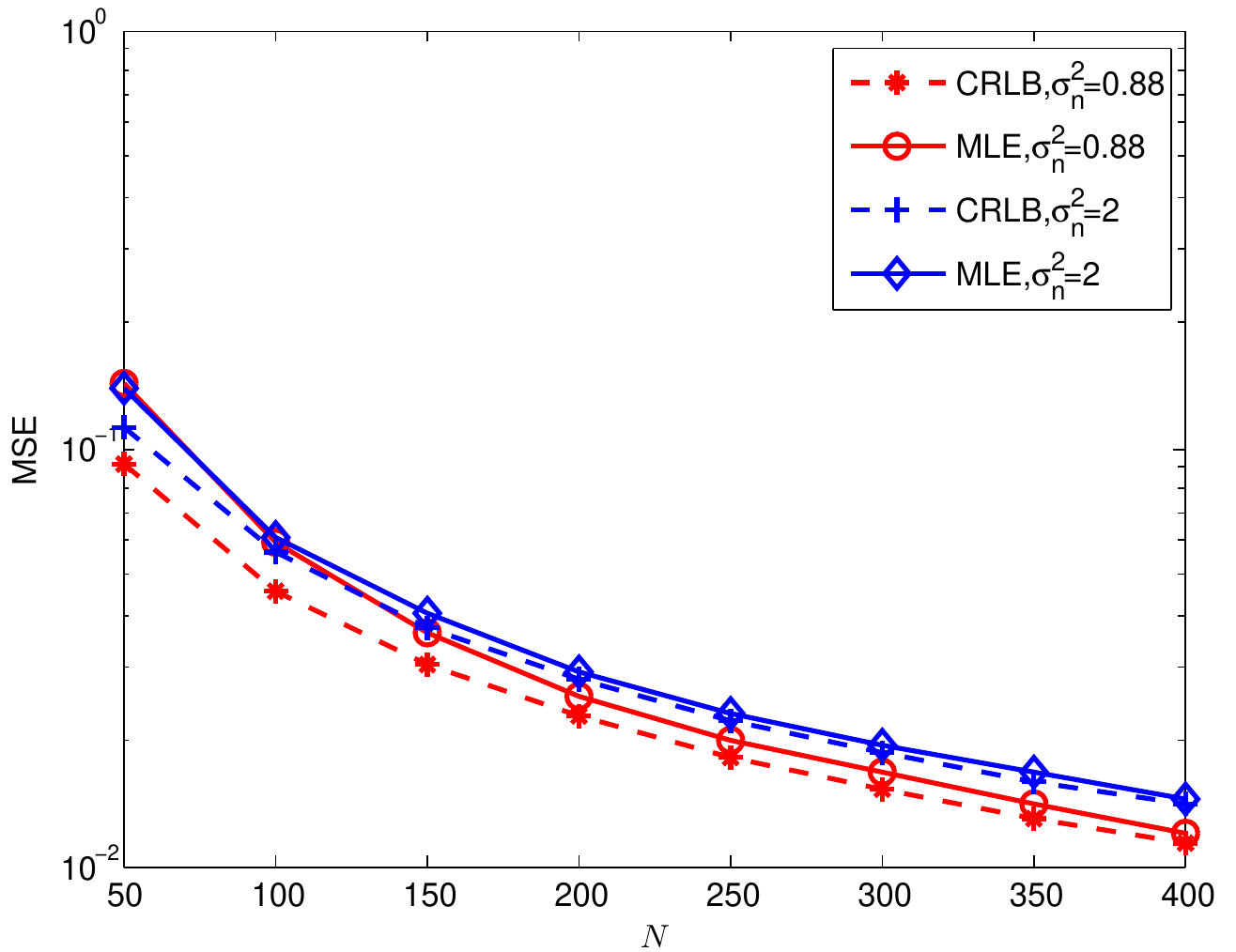}
\caption{The MSE of the ML estimator and the CRLB for different number of observations when $\sigma_e^2$ is fixed as $0.3$.}
\label{fig_mse3}
\end{figure}

Finally, the ML estimator is compared with other estimation methods. The parameters are set as follows: $\sigma_n^2=1$, $\sigma_e^2=0.4$. Three estimators are considered, including the ML estimator (\ref{ML}), the perturbation-ignored estimator (\ref{opt_a}) and a perturbation-known estimator corresponding to a completely known sensing matrix. All the MSEs are then compared with the CRLB (\ref{mse_crlb}). The results are plotted in Fig. \ref{Comp_vector}. It is obvious that the perturbation-known case performs better than the CRLB in which the perturbation is assumed unknown. When the number of measurements $N$ is smaller than $300$, the perturbation-ignored estimator works better than the ML estimator. 
The reason is that the norm of the ML estimator has a much larger probability of being the norm threshold $R_w$ than that of the perturbation-ignored estimator when the number of measurements is not large enough. While the effect of the bias of the perturbation-ignored estimator becomes apparent when the number of measurements increases.   
Thus the MSE of the perturbation-ignored estimator decreases slowly as $N$ increases. In fact, the MSE converges to $0.031$ according to (\ref{MSE_converge}). The MSE of the ML estimator decreases as expected, and asymptotically achieves the CRLB.
\begin{figure}[t]
\centering
\includegraphics[width=3.5in]{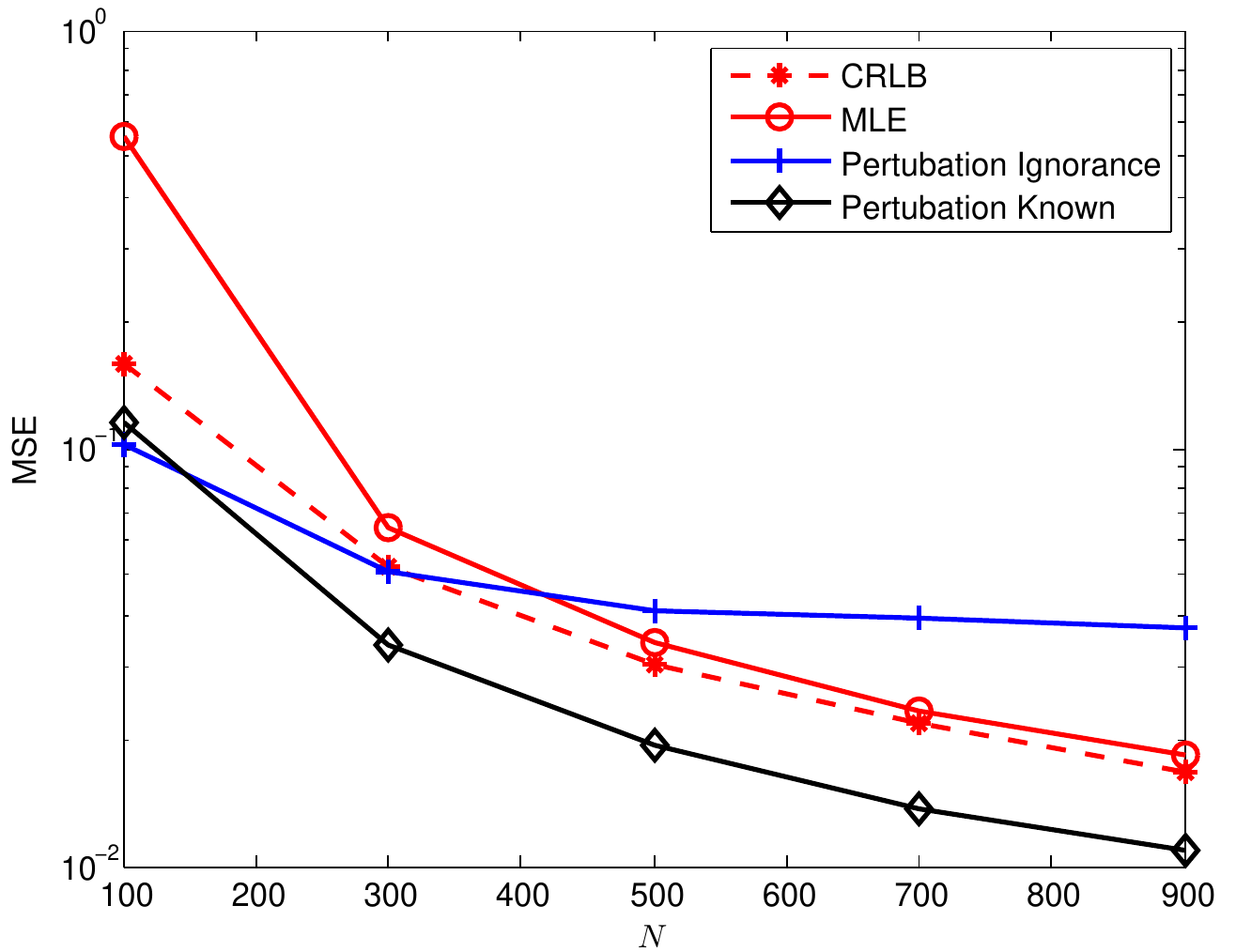}
\caption{MSE comparison in estimating the vector parameter by three estimators.}
\label{Comp_vector}
\end{figure}
\section{Conclusion}
In this paper, we have studied the problem of estimating a deterministic parameter vector from sign measurements with a perturbed sensing matrix. Firstly, the ML estimator was utilized to estimate the unknown parameter and it was proved to be consistent. The CRLB was derived to analyze the performance of the estimator. It was demonstrated that with the variance of the equivalent noise fixed, the perturbation exacerbates the performance of the estimation. Meanwhile, under certain relationship between the variance of additive noise and the strength of the perturbation, the CRLB on the MSE will achieve its minimum. This result demonstrates that in the MSE sense, suitable perturbation may be beneficial in some special cases. Secondly, it was shown that the original ML estimation problem could be transformed to a convex optimization problem, which can be efficiently solved. Theoretical analysis implied that suitable perturbation may be beneficial to improve the probability that the optimal point of the ML estimator exists. Furthermore, under a perturbed sensing matrix, the perturbation-ignored estimator is a scaled version with the same direction of the ML estimator. It was also shown that the perturbation-ignored estimator works well when the number of measurements is not large enough and the perturbation is small. However, the perturbation-ignored estimator is biased and its MSE converges to a constant as the number of measurements increases. In contrast, the ML estimator is unbiased and achieves the CRLB in the asymptotic sense.

\appendices

\section{Proof of Theorem~\ref{ML_con}}
\begin{proof}\label{MLcon_appendix}
We first define the normalized log-likelihood function as
\begin{align}\label{lN}
l_N({\mathbf w})&\triangleq\frac{1}{N} l({\mathbf y};{\mathbf w})=\frac{1}{N}\sum\limits_{i=1}^{N}{\log }\Phi \left(y_i\frac{{{\mathbf h}_i^{\rm T}{\mathbf w}}}{\sigma_z}\right).
\end{align}

As $N$ tends to infinity, the weak law of large numbers implies that
\begin{align}
\lim_{N\rightarrow \infty}l_N({\mathbf w})\stackrel{\rm p}{\rightarrow}l_0({\mathbf w})\triangleq{\rm E}_{y,{\mathbf h}}\left[\log\Phi \left(y\frac{{{\mathbf h}^{\rm{T}}\mathbf{w}}}{\sigma_z}\right)\right],\label{L0}
\end{align}
where the expectation is taken with respect to $y$ and $\mathbf h$, and the notation $\stackrel{\textrm p}{\rightarrow}$ denotes convergence in probability. Since model (\ref{Gaussian_model}) is identifiable, $l_0({\mathbf w})$ has a unique maximum attained at ${\mathbf w}_0$ by the information inequality \cite{Kay211}. In order to claim that $\hat{\mathbf w}_{\rm ML}$ converges to ${\mathbf w}_0$ in probability as $N\rightarrow \infty$, one needs to ensure that the limiting and maximization operations in (\ref{gen_p}) and (\ref{L0}) can be interchanged. Sufficient conditions for the maximum of the limit to be the limit of the maximum are that the parameter space is compact and the normalized log-likelihood function $l_N({\mathbf w})$ converges uniformly to $l_0({\mathbf w})$ in probability \cite{WNewey}. It is obvious that the parameter space $\mathcal W$ is bounded and closed. To prove the uniform convergence in probability, note that $l_N({\mathbf w})$ is continuous, thus it suffices to show that there exists a function $U({\mathbf y}, {\mathbf H})$ such that
\begin{align}
|l_N({\mathbf w})|\leq U({\mathbf y}, {\mathbf H}),\quad \forall ~{\mathbf w}\in {\mathcal W}.\label{findfun}
\end{align}

To find such a function $U({\mathbf y}, {\mathbf H})$, we may use the mean value expansion of $q({\mathbf w}) = \log\Phi\left(y_i\frac{\textstyle{{\mathbf h}_i^{\rm T}{\mathbf w}}}{\textstyle{\sigma_z}}\right)$ around the origin ${\mathbf w}={\mathbf 0}$. Notice that the derivative of $\log\Phi(x)$ is
\begin{align}
k(x)\triangleq\frac{\partial\log\Phi(x)}{\partial x}=\frac{1}{{\Phi(x)}}\frac{{\partial\Phi(x)}}{\partial x},\notag
\end{align}
which is convex and positive. When $x\rightarrow \infty$, $k(x)$ tends to zero, on the other hand, $k(x)$ tends to $-x$ as $x\rightarrow -\infty$. As a consequence, there exists a suitable constant $C>0$ such that
\begin{align}\notag
k(x)\leq C(1+|x|).
\end{align}
By the mean value theorem, the following result is obtained,
\begin{align}\label{mlcons1}
\left|\log\Phi\left(y_i\frac{{{\mathbf h}_i^{\rm T}{\mathbf w}}}{{\sigma_z}}\right)\right|&=\left|\log\Phi(0)+\nabla q({\mathbf w}')^{\rm T}{\mathbf w}\right|\notag\\
&\leq \left|\log\Phi(0)\right|+\left\|\nabla q({{\mathbf w}'})\right\|_2\left\|{\mathbf w}\right\|_2,
\end{align}
where ${\mathbf w}'$ is some point lying in the parameter space $\mathcal W$. The norm of the gradient $\nabla q({\mathbf w}')$ can be upper bounded by
\begin{align}\label{mlcons2}
\left\|\nabla q\left({{\mathbf w}'}\right)\right\|_2&=\left\|\frac{y_i}
{{\tilde{\sigma}_z}}k\left(y_i\frac{{{\mathbf h}_i^{\rm T}{{\mathbf w}'}}}{{\tilde{\sigma}_z}}\right)\left({\mathbf h}_i-\frac{\sigma_e^2}{{\tilde{\sigma}_z^2}}\left({\mathbf h}_i^{\rm T}{{\mathbf w}'}\right){{\mathbf w}'}\right)\right\|_2\notag\\
&\leq\frac{1}{\sigma_n}k\left(y_i\frac{{{\mathbf h}_i^{\rm T}{\mathbf w}'}}{{\tilde{\sigma}_z}}\right)\left\|{\mathbf h}_i-\frac{{\sigma_e^2}}{{\tilde{\sigma}_z^2}}\left({\mathbf h}_i^{\rm T}{\mathbf w}'\right){\mathbf w}'\right\|_2\notag\\
&\leq\frac{C}{{\sigma}_n}\left(1+\left|\frac{{{\mathbf h}_i^{\rm T}{\mathbf w}'}}{{\tilde{\sigma}_z}}\right|\right)\left\|{\mathbf h}_i-\frac{{\sigma_e^2}}{{\tilde{\sigma}_z^2}}\left({\mathbf h}_i^{\rm T}{{\mathbf w}'}\right){{\mathbf w}'}\right\|_2\notag\\
&\leq\frac{C}{{\sigma}_n}\left(1+\frac{{R_w}}{\sigma_n}\left\|{\mathbf h}_i\right\|_2\right)\left\|{\mathbf h}_i-\frac{{\sigma_e^2}}{{\tilde{\sigma}_z^2}}\left({\mathbf h}_i^{\rm T}{{\mathbf w}'}\right){{\mathbf w}'}\right\|_2,
\end{align}
where
\begin{align}\label{mlcons3}
\left\|{\mathbf h}_i-\frac{{\sigma_e^2}}{{\tilde{\sigma}_z^2}}\left({\mathbf h}_i^{\rm T}{{\mathbf w}'}\right){\mathbf w}'\right\|_2&\leq\left(\left\|{\mathbf h}_i\right\|_2+\frac{{\sigma_e^2}}{{{\sigma}_n^2}}{R_w}\left\|{\mathbf h}_i\right\|_2\left\|{\mathbf w}'\right\|_2\right)\notag\\
&\leq\left(1+\frac{{\sigma_e^2}}{{{\sigma}_n^2}}R_w^2\right)\left\|{\mathbf h}_i\right\|_2.
\end{align}
The above inequalities follow from $\|{\mathbf w}\|_2\leq R_w$, $\|{{\mathbf w}'}\|_2\leq R_w$, and $\tilde{\sigma}_z^2 = \|{\mathbf w}'\|_2^2\sigma_e^2+\sigma_n^2\geq \sigma_n^2$. Plugging (\ref{mlcons1}), (\ref{mlcons2}) and (\ref{mlcons3}) into (\ref{lN}), it follows that
\begin{align}
l_N({\mathbf w})\leq |\log\Phi(0)|+\frac{C_1}{N}\sum_{i=1}^N\left(1+\frac{\textstyle{R_w}}{{{\sigma}_n}}\|{\mathbf h}_i\|_2\right)\|{\mathbf h}_i\|_2,\label{bound}
\end{align}
where $C_1$ is a constant. Using $U({\mathbf y}, {\mathbf H})$ to denote the right side of the equation (\ref{bound}), the condition (\ref{findfun}) is satisfied. Therefore, the consistency of the  ML estimator is proved.
\end{proof}

\section{Proof of Theorem \ref{crlb_theorem}}\label{CRB_appendix}

\begin{proof}
We first show that the regularity condition holds for the likelihood function ${\rm {Pr}}(\mathbf y;\mathbf w)$. The gradient of the log-likelihood function $l(\mathbf{y};\mathbf{w})$ with respect to $\mathbf w$ is
\begin{align}
    \nabla_{\mathbf w}l({\mathbf y};{\mathbf w})\notag
    =\frac{1}{\sqrt{2\pi}\sigma_z}\sum_{i=1}^N
    \Bigg(\frac{y_i}{\Phi \left(y_i\frac{{\mathbf h}_i^{\rm T}{\mathbf w}}{\sigma_z} \right)}{\textrm e}^{-\frac{{\left({\mathbf h}_i^{\rm T}{\mathbf w} \right)}^2}{2\sigma_z^2}}\left({\mathbf h}_i-\frac{\sigma_e^2}{\sigma_z^2}({\mathbf h}_i^{\rm T}{\mathbf w}){\mathbf w} \right)\Bigg).\notag
\end{align}
The probability distribution function for $y_i$ is
\begin{align}
    y_i =
    \begin{cases}
       -1,&{\rm {with~probability}}\quad\Phi \left(-\frac{{\mathbf h}_i^{\rm T}{ \mathbf w}}{\sigma_z} \right); \\
       1,&{\rm {with~probability}}\quad \Phi \left(\frac{{\mathbf h}_i^{\rm T}{\mathbf w}}{\sigma_z}\right).\nonumber
   \end{cases}
\end{align}
It follows that for all $\mathbf w$, the regularity condition ${\rm E}_{\mathbf y}\left[\nabla_{\mathbf w}l({\mathbf y};{\mathbf w})\right]={\mathbf 0}$ holds.

Fortunately, a closed-from expression for the CRLB can be obtained in the case of a vector parameter CRLB for transformation.

Suppose that we wish to estimate ${\boldsymbol {\alpha}}={\mathbf g}({\boldsymbol \theta})$, where $\mathbf g$ is a $r$-dimensional function and $\boldsymbol \theta$ is a $s$-dimensional parameter vector. Then the CRLB of ${\boldsymbol \alpha}$ from ${\boldsymbol \theta}$ is given by \cite{Kay45}
\begin{align}
{\rm {Cov}}(\hat{\boldsymbol \alpha})\succeq \frac{\partial{\mathbf g}({\boldsymbol \theta})}{\partial{\boldsymbol \theta}}\left({\mathbf J}(\boldsymbol \theta)\right)^{-1}\frac{\partial{\mathbf g}({\boldsymbol \theta})^{\rm T}}{\partial{\boldsymbol \theta}}, \label{CRB_lemma}
\end{align}
where ${\mathbf J}(\boldsymbol \theta)$ is the FIM of $\boldsymbol \theta$.
%\end{lemma}
%The CRLB for $\mathbf w$ can also be derived by using the vector transformation of parameters.
%where $\frac{\partial{\mathbf g}({\mathbf \theta})}{\partial{\mathbf \theta}}$ is the corresponding $r\times s$ Jacobian matrix.

We may define
\begin{align}
{\mathbf v} = \frac{\mathbf w}{\sqrt{\|\mathbf w\|_2^2\sigma_e^2+\sigma_n^2}}. \notag
\end{align}
%The relationship between $\tilde{\mathbf w}$ and ${\mathbf w}$ is a one to one mapping.
Then $\mathbf w$ can be uniquely determined from ${\mathbf v}$ by
\begin{align}
{\mathbf w}=\frac{\sigma_n}{{\sqrt{1-\sigma_e^2\|{\mathbf v}\|_2^2}}}{{\mathbf v}}.\notag
\end{align}

In our setting, ${\boldsymbol \alpha}={\mathbf w}$, ${\boldsymbol \theta}={{\mathbf v}}$ and ${\mathbf w}={\mathbf g}({{\mathbf v}})$. The log-likelihood function $l({\mathbf y};{\mathbf v})$ for ${\mathbf v}$ is
\begin{align}
 l({\mathbf y};{\mathbf v})=\sum\limits_{i=1}^{N}{\log}\Phi\left(y_i{\mathbf h}_i^{\rm{T}}{\mathbf v}\right).\notag
\end{align}
Its gradient and Hessian are
\begin{align}\notag
\nabla_{{\mathbf v}}l({\mathbf y};{\mathbf v})=\frac{1}{\sqrt{2\pi}}\sum_{i=1}^N
    \frac{y_i}{\Phi \left(y_i{\mathbf h}_i^{\rm T}{\mathbf v}\right)}{\textrm e}^{-\frac{\left({\mathbf h}_i^{\rm T}{\mathbf v} \right)^2}{2}}{\mathbf h}_i,
\end{align}
and
\begin{align}
\nabla_{{\mathbf v}}^2l({\mathbf y};{\mathbf v})=&-\frac{1}{\sqrt{2\pi}}\sum_{i=1}^N
    \frac{y_i}{\Phi \left(y_i{\mathbf h}_i^{\rm T}{\mathbf v}\right)}{\rm e}^{-\frac{\left({\mathbf h}_i^{\rm T}{\mathbf v} \right)^2}{2}}({\mathbf h}_i^{\rm T}{{\mathbf v}}){\mathbf h}_i{\mathbf h}_i^{\rm T}\notag\\
&-\frac{1}{2\pi}\sum_{i=1}^N
    \frac{1}{\Phi^2\left(y_i{\mathbf h}_i^{\rm T}{\mathbf v}\right)}{\rm e}^{-{\left({\mathbf h}_i^{\rm T}{\mathbf v} \right)^2}}{\mathbf h}_i{\mathbf h}_i^{\rm T},\label{Hessian_tf}
\end{align}
respectively. The FIM can be computed as
%\begin{align}\label{FIM_tf}
%{\mathbf J}({\tilde{\mathbf w}})=-{\rm E}_{\mathbf y}[\nabla_{\tilde{\mathbf w}}^2l({\mathbf y};\tilde{\mathbf w})].
%\end{align}
%Substituting (\ref{Hessian_tf}) in (\ref{FIM_tf}), one has
\begin{align}
{\mathbf J}({{\mathbf v}}) = \sigma_z^2{\mathbf H}{\mathbf \Lambda}{\mathbf H}^{\rm T},\label{J1}
\end{align}
where $\mathbf \Lambda$ is defined as (\ref{lambda_ii}). The corresponding Jacobian matrix is
\begin{align}
\frac{\partial{\mathbf g}({{\mathbf v}})}{\partial{{{\mathbf v}}}}=\sigma_z({\mathbf I}+\frac{\sigma_{e}^2}{\sigma_{n}^2}{\mathbf w}{\mathbf w}^{\rm T}).\label{Jacob}
\end{align}
%Substituting (\ref{J1}) and (\ref{Jacob}) in (\ref{CRB_lemma}), the CRLB is obtained by
%\begin{align}\label{CRLB_new}
%{\rm {Cov}}({\hat{\mathbf w}})\succeq \left({\mathbf I}+\frac{\sigma_{e}^2}{\sigma_{n}^2}{\mathbf w}{\mathbf w}^{\rm T}\right)({\mathbf H}{\mathbf \Lambda}{\mathbf H}^{\rm T})^{-1}\left({\mathbf I}+\frac{\sigma_{e}^2}{\sigma_{n}^2}{\mathbf w}{\mathbf w}^{\rm T}\right).
%\end{align}
By employing the Sherman-Morrison formula \cite{Golub} and (\ref{Var_z}), one has
\begin{align}
\left({\mathbf I}-\frac{\sigma_{e}^2}{\sigma_{z}^2}{\mathbf w}{\mathbf w}^{\rm T}\right)^{-1}={\mathbf I}+\frac{\sigma_{e}^2}{\sigma_{n}^2}{\mathbf w}{\mathbf w}^{\rm T}.\label{bianhuan}
\end{align}
Substituting (\ref{J1}), (\ref{Jacob}) and (\ref{bianhuan}) in (\ref{CRB_lemma}), the CRLB is
\begin{align}
{\rm {Cov}}({\hat{\mathbf w}})\succeq \left({\mathbf M}{\mathbf \Lambda}{\mathbf M}^{\rm T}\right)^{-1},\nonumber
\end{align}
where $\mathbf M$ is defined as (\ref{Def_M}). The MSE is equal to the trace of the covariance matrix. Therefore, the result (\ref{mse_crlb}) is established.
\end{proof}
\section{Proof of Proposition~\ref{Pro_ineq}}
\begin{proof}\label{Ineq_appendix}
We first define
\begin{align}
{\mathbf N}={\mathbf I}-\frac{\sigma_e^2}{\sigma_z^2}{\mathbf w}{\mathbf w}^{\rm T}.\notag
\end{align}
Since $\tilde{\mathbf J}$ and $\mathbf N$ are both positive definite matrices, they can be factored as $\tilde{\mathbf J} = \mathbf U \tilde{\mathbf \Delta}{\mathbf U}^{\rm T}$ and ${\mathbf N} = {\mathbf V}{\mathbf \Delta}_{\mathbf N}{\mathbf V}^{\rm T}$, where ${\mathbf U},{\mathbf V} \in \mathbb{R}^{p\times p}$ are both orthogonal matrices, and $\tilde{\mathbf \Delta}={\rm {diag}}(\tilde{\lambda}_1,\cdots,\tilde{\lambda}_p)$, ${\mathbf \Delta}_{\mathbf N}={\rm {diag}}\left(1,\cdots,1,\frac{\textstyle {1}}{\textstyle {1+\gamma}}\right)$. Note that all eigenvalues except the last one are equal to $1$.
Using the equality ${\rm{tr}}(\mathbf {AB})={\rm{tr}}(\mathbf {BA})$, we obtain
\begin{align}
{\rm {tr}}(\tilde{\mathbf J})={\rm {tr}}(\tilde{\mathbf \Delta}^{-1})=\sum_{i=1}^{p}\frac{1}{\tilde{\lambda}_i}.\notag
\end{align}
When the variance of the equivalent noise $\sigma_z^2$ is fixed, one has
\begin{align}
{\mathbf J}&={\mathbf N}{\tilde{\mathbf J}}{\mathbf N}^{\rm T}\notag\\
           &={\mathbf V}{\mathbf \Delta}_{\mathbf N}{\mathbf V}^{\rm T}{\mathbf U}\tilde{\mathbf \Delta}{\mathbf U}^{\rm T}{\mathbf V}{\mathbf \Delta}_{\mathbf N}{\mathbf V}^{\rm T}.\notag
\end{align}
It follows that
\begin{align}
{\rm {tr}}({\mathbf J}^{-1})={\rm {tr}}({\mathbf \Delta}_{\mathbf N}^{-1}{\mathbf V}^{\rm T}{\mathbf U}\tilde{\mathbf \Delta}{\mathbf U}^{\rm T}{\mathbf V}{\mathbf \Delta}_{\mathbf N}^{-1}).\notag
\end{align}
Defining ${\mathbf Q}={\mathbf V}^{\rm T}{\mathbf U}\tilde{\mathbf \Delta}{\mathbf U}^{\rm T}{\mathbf V}$ and ${\mathbf T} ={\mathbf \Delta}_{\mathbf N}^{-1}{\mathbf \Delta}_{\mathbf N}^{-1}$, it is obvious that $\mathbf Q$ has the same eigenvalues with $\tilde{\mathbf \Delta}$. Since both $\mathbf Q$ and ${\mathbf T}$ are positive definite, by using the trace inequality \cite{Lasserre}, we have
\begin{align}
    \sum_{i=1}^{n}{\lambda_{{\mathbf Q},i}{\lambda_{{\mathbf T},n-i}}}\leq {\rm{tr}}({\mathbf {QT}})\leq
     \sum_{i=1}^{n}{\lambda_{{\mathbf Q},i}{\lambda_{{\mathbf T},i}}}.\notag
\end{align}
Therefore, the desired result (\ref{inequ}) is obtained.
\end{proof}

\section{Proof of Proposition \ref{uniq}}

\begin{proof}\label{Strict_convex}
We use $f({\mathbf v})$ to denote the objective function of (\ref{OPT}). According to (\ref{Hessian_tf}), the Hessian of $f({\mathbf v})$ is $$\nabla_{\mathbf v}^2 f({\mathbf v})=\sum_{i=1}^{N}\beta_i{\mathbf h}_i{\mathbf h}_i^{\rm T},$$ where
\begin{align}
\beta_i=\frac{1}{\sqrt{2\pi}\Phi^2\left(y_ix_i\right)}{\rm e}^{-\frac{1}{2}x_i^2}\left(\frac{1}{\sqrt{2\pi}}{\rm e}^{-\frac{1}{2}x_i^2}+y_ix_i\Phi(y_ix_i)\right),\notag
\end{align}
and $x_i={\mathbf h}_i^{\rm T}{\mathbf v}$. By the inequality $x\Phi(-x)<\frac{1}{\sqrt{2\pi}}{\rm e}^{-\frac{1}{2}x^2}, \forall~x\in \mathbb{R}$, one can show that $\beta_i>0$. Thus $\nabla_{\mathbf v}^2 f({\mathbf v})\succ 0$. The result is established.
\end{proof}

\section{Proof of Proposition \ref{Pr_b}}

\begin{proof}\label{Prob_appendix}
%We first define a new variable
% \begin{align}
%% \nonumber to remove numbering (before each equation)
%v = \frac{\textstyle{w}}{\textstyle{\sqrt{|w|^2\sigma_e^2+\sigma_n^2}}}.\notag
%\end{align}The above optimization problem has an analytic form.
We first solve the unconstrained optimization problem
\begin{align}
v_u^*=\underset{v\in \mathbb {R}} {\operatorname{argmin}}~-\sum_{i=1}^{N}{\log }\Phi \left(y_iv\right).\notag
\end{align}
Assuming that the observation $\{y_i\}_{i=1}^N$ has $k$ ones. Setting the derivative of the objective function to zero and using the equality $\Phi (v_u^*)+\Phi (-v_u^*)=1$, one has
%To obtain the result, we construct a new problem as follows
%\begin{align}
%z_i={\rm {sign}}(v+n_i),\quad i=1,2,\cdots,N,\notag
%\end{align}
%where $n_i\sim {\mathcal N}(0,1)$. By simple derivation, $v^*$ is the ML estimator of $v$ with observation $z_i=y_i, i=1,\cdots,N$. The probability mass function for $z_i$ is
%\begin{align}
%    z_i =
%    \begin{cases}
%       -1,&{\rm {with~probablity}}\quad \Phi (-v); \\
%       1,&{\rm {with~probablity}}\quad \Phi (v).\notag
%   \end{cases}
%\end{align}
%Thanks to the invariance principle of the ML estimation, we have
%\begin{align}
%{\rm E}_{z_i}(z_i)&=\Phi (v)-\Phi (-v)\notag\\
%\end{align}
%Using the equality $\Phi (v^*)+\Phi (-v^*)=1$, we have
\begin{align}
v_u^*=\Phi^{-1}\left(\frac{k}{N}\right) ,\notag
\end{align}
where $\Phi^{-1}$ denotes the inverse function of $\Phi$.
Now we calculate the probability ${\rm P}_{\mathcal V}$ as
\begin{align}\label{bounddd}
{\rm P}_{\mathcal V}&={\rm {Pr}}\left[\left|v_u^*\right|<\frac{1}{\sigma_e}\right]\notag\\
&={\rm {Pr}}\left[-\frac{1}{\sigma_e}<\Phi^{-1}\left(\frac{k}{N}\right)<\frac{1}{\sigma_e}\right]\notag\\
    &={\rm {Pr}}\left[N\Phi\left(-\frac{1}{\sigma_e}\right)<k<N\Phi\left(\frac{ 1}{\sigma_e}\right)\right].\notag
\end{align}
%\begin{align}\label{bounddd}
%{\rm P}_{\mathcal V}&={\textrm {Pr}}\left[\left|v^*\right|<\frac{1}{\sigma_e}\right]\notag\\
%&={\rm {Pr}}\left[-\frac{1}{\sigma_e}<\Phi^{-1}\left(\frac{1}{2}
%    +\frac{1}{2N}\sum_{i=1}^Ny_i\right)<\frac{1}{\sigma_e}\right]\notag\\
%&={\rm {Pr}}\left[\Phi\left(-\frac{1}{\sigma_e}\right)<\frac{1}{2}
%    +\frac{1}{2N}\sum_{i=1}^Ny_i<\Phi\left(\frac{1}{\sigma_e}\right)\right],
%\end{align}
%\begin{align}\label{bounddd}
%{\rm P}_{b}&={\rm {Pr}}\left[\left|t^*\right|> R_v\right]\notag\\
%&=1-{\rm {Pr}}\left[-R_v\leq \Phi^{-1}\left(\frac{1}{2}
%    +\frac{1}{2N}\sum_{i=1}^Ny_i\right)\leq R_v\right]\notag\\
%&=1-{\rm {Pr}}\left[\Phi\left(-R_v\right)\leq\frac{1}{2}
%    +\frac{1}{2N}\sum_{i=1}^Ny_i\leq\Phi\left(R_v\right)\right],
%\end{align}
where the last step follows from the monotone increasing property of $\Phi^{-1}$.
%
%Assuming the observation $\{y_i\}_{i=1}^N$ has $k$ ones. Then
%\begin{align}\label{kN}
%\frac{1}{2}+\frac{1}{2N}\sum_{i=1}^Ny_i=\frac{k}{N}.
%\end{align}
%Plugging (\ref{kN}) into (\ref{bounddd}), we can obtain
%\begin{align}
%{\rm P}_{\mathcal V}
%    %&={\rm {Pr}}\left[\Phi\left(-\frac{1}{\sigma_e}\right)<\frac{1}{2}
%    %+\frac{1}{2N}(2k-N)<\Phi\left(\frac{1}{\sigma_e}\right)\right]\notag\\
%    &={\rm {Pr}}\left[\Phi\left(-\frac{1}{\sigma_e}\right)<\frac{k}{N}<\Phi\left(\frac{1}{\sigma_e}\right)\right]\notag\\
%    &={\rm {Pr}}\left[N\Phi\left(-\frac{1}{\sigma_e}\right)<k<N\Phi\left(\frac{ 1}{\sigma_e}\right)\right].\notag
%\end{align}
%Define ${k_{\rm min}}=\left\lfloor N\Phi\left(-\frac{1}{\sigma_e}\right)\right\rfloor+1$ and ${k_{\rm max}}=\left\lceil N\Phi\left(\frac{1}{\sigma_e}\right)\right\rceil-1$.
Since
\begin{align}
    y_i=
    \begin{cases}
       -1,&{\rm with~probability}\quad \Phi \left(-\frac{w_0}{\sigma_z} \right); \\
       1,&{\rm with~probability}\quad \Phi \left(\frac{w_0}{\sigma_z}\right),\notag
   \end{cases}
\end{align}
%Define ${k_{min}}=\left\lceil N\Phi\left(-R_v\right)\right\rceil$ and ${k_{max}}=\left\lfloor N\Phi\left(R_v\right)\right\rfloor$. Since
%\begin{align}
%    y_i \sim
%    \begin{cases}
%       -1,&{\rm {with~probablity}}\quad \Phi \left(-\frac{\textstyle w_0}{\textstyle \sigma_z} \right); \\
%       1,&{\rm {with~probablity}} \quad \Phi \left(\frac{\textstyle w_0}{\textstyle \sigma_z}\right),
%   \end{cases}
%\end{align}
the result (\ref{Cal_pro}) is established.
\end{proof}

\end{document}